\newif\ifieee
\ieeetrue

\ifieee
\documentclass[reprint,showkeys,amsmath,amssymb,amsthm,aps,letterpaper]{IEEEtran}
\else
\documentclass[10pt,a4paper]{amsart}
\fi

\usepackage{latexsym}
\usepackage{enumerate}
\usepackage{amssymb}
\usepackage{amsmath}
\usepackage{amsthm}
\usepackage{mathtools}
\usepackage{xcolor}
\usepackage{blkarray}
\usepackage{bbm}
\usepackage%
{hyperref}
\hypersetup{
  colorlinks,
  bookmarksopen,
  bookmarksnumbered,
  citecolor=blue,
  linkcolor=black,
  pdfstartview=FitH,
}

\newtheorem{theorem}{Theorem}[section]
\newtheorem{lemma}[theorem]{Lemma}
\newtheorem{proposition}[theorem]{Proposition}
\newtheorem{corollary}[theorem]{Corollary}

\theoremstyle{definition}

\newtheorem{definition}[theorem]{Definition}

\newtheorem{remark}[theorem]{Remark}

\newcommand{\qand}{\quad\text{and}\quad}

\newcommand\tr{\mathop{\rm tr}}

\let\epsilon\varepsilon

\newenvironment{sbmatrix}{\left[\begin{smallmatrix}}{\end{smallmatrix}\right]}

\newcommand{\thetatilde}{\wt{\vartheta}}

\DeclareMathOperator{\mpsdLetter}{m}
\newcommand{\mpsd}[1][]{\mpsdLetter\ensuremath{{}^{\kern -.5ex +}_{\kern-.25ex #1}}}
\DeclareMathOperator{\diag}{diag}
\DeclareMathOperator{\spn}{span}
\DeclareMathOperator{\inter}{int}
\DeclareMathOperator{\plex}{plex}

\DeclareMathOperator{\qinter}{qint}
\DeclareMathOperator{\rank}{rank}
\DeclareMathOperator{\cc}{int}
\DeclareMathOperator{\iso}{iso}
\DeclareMathOperator{\msr}{msr}

\newcommand{\cl}[1]{\mathcal{#1}}
\newcommand{\bb}[1]{\mathbb{#1}}

\def\dc#1{\expandafter\def\csname#1\endcsname{\mathcal{#1}}}
\def\db#1{\expandafter\def\csname b#1\endcsname{\mathbb{#1}}}
\def\loopy#1#2{%
  \def#1##1{\def\next{#2{##1}#1}\ifx##1\relax\let\next\relax\fi\next}}
\loopy{\makemathcals}{\dc}\loopy{\makemathbbs}{\db}
\makemathbbs  CRNQZTDFPE\relax
\makemathcals QWERTYUIOPASDFGHJKLZXCVBNM\relax

\makeatletter
\newcommand*\wt[1]{\mathpalette\wthelper{#1}}
\newcommand*\wthelper[2]{%
        \hbox{\dimen@\accentfontxheight#1%
                \accentfontxheight#11.2\dimen@
                $\m@th#1\widetilde{#2}$%
                \accentfontxheight#1\dimen@
        }%
}

\newcommand*\accentfontxheight[1]{%
        \fontdimen5\ifx#1\displaystyle
                \textfont
        \else\ifx#1\textstyle
                \textfont
        \else\ifx#1\scriptstyle
                \scriptfont
        \else
                \scriptscriptfont
        \fi\fi\fi3
}
\makeatother

\begin{document}

\title{Complexity and capacity bounds for quantum channels} 

\author{\IEEEauthorblockN{Rupert H. Levene\IEEEauthorrefmark{1}, Vern I. Paulsen\IEEEauthorrefmark{1}, and Ivan G. Todorov\IEEEauthorrefmark{3}}\\[6pt]
  \IEEEauthorblockA{\IEEEauthorrefmark{1}School of Mathematics and Statistics, University College Dublin, Belfield, Dublin~4, Ireland}\\
  \IEEEauthorblockA{\IEEEauthorrefmark{2}Institute for Quantum Computing and Dept.~of Pure Math., University of Waterloo, Waterloo, Ontario, Canada}\\
  \IEEEauthorblockA{\IEEEauthorrefmark{3}Mathematical Sciences Research Centre, Queen's University Belfast, Belfast BT7 1NN, United Kingdom}}
    
\date{17 October 2017}

\maketitle

\begin{abstract}
  We generalise some well-known  graph parameters to
  operator systems
  by considering their underlying quantum channels.
  In particular, we introduce the quantum complexity as the dimension
  of the smallest co-domain Hilbert space a quantum channel requires
  to realise a given operator system as its non-commutative
  confusability graph. We describe quantum complexity as a generalised
  minimum semidefinite rank and, in the case of a graph operator
  system, as a quantum intersection number.  The quantum complexity
  and a closely related quantum version of orthogonal rank turn out to
  be upper bounds for the Shannon zero-error capacity of a quantum
  channel, and we construct examples for which these bounds beat the best
  previously known general upper bound for the capacity of quantum
  channels, given by the quantum Lov\'asz theta number.
\end{abstract}

\section{Introduction}\label{s_intro}

In 1956, Shannon \cite{shannon} initiated zero-error information theory, 
introducing the concept of the \emph{confusability graph} $G_{\cl N}$ of a finite input-output information channel 
$\cl N:X \to Y$. 
The vertex set of $G_{\cl N}$ is the input alphabet $X$, 
and two vertices form an edge in $G_{\cl N}$ if they can result in the same symbol from the output alphabet $Y$ 
after transmission via $\cl N$. 
Shannon showed that the zero-error behaviour of $\cl N$, and various measures of its capacity, 
depend only on the graph $G_{\cl N}$.  
In particular, if two channels have the same confusability graphs, then they have the same one-shot zero-error capacity and the same 
Shannon capacity. It is not hard to see that every graph with vertex set $X$ is the confusability graph of some---and in fact many---information channels. Thus, a natural measure of the \emph{complexity} of a graph 
$G$ on $X$ is the minimal cardinality of $Y$ over all realisations of $G$ as the confusability graph of a channel $\cl N:X \to Y$. 

Similarly, Duan, Severini and Winter \cite{dsw} showed that every quantum channel~$\Phi:M_n \to M_k,$ 
where $M_m$ denotes the set of complex $m \times m$ matrices, has an associated \emph{non-commutative} confusability
graph $\S_\Phi$, which they defined as a certain operator subsystem of $M_n$. 
In the case $\Phi$ is a classical channel, $\cl S_{\Phi}$ coincides with the \emph{graph operator system}
of $G_{\Phi}$ (see \cite[equation~(3)]{dsw}). 
As in Shannon's case, they proved that many natural measures of the quantum capacity of such a channel depend only 
on the operator system $\cl S$. It is again the case that every operator subsystem of $M_n$ 
arises as the non-commutative confusability graph of 
potentially many quantum channels. \ifieee\goodbreak\noindent\fi  Thus, we are lead  
to define the {\it quantum complexity} $\gamma(\cl S)$ 
of an operator subsystem $\cl S$ of $M_n$ as the least positive integer 
$k$ for which there exists 
a quantum channel $\Phi: M_n \to M_k$ such that $\cl S = \cl S_{\Phi}$. 

The goal of this paper is to study this and other, closely related, measures of complexity 
and  to derive their relationships with various measures of  capacity for classical and quantum channels. 
We will show, in particular, that the measures of complexity we introduce give upper bounds on the zero-error capacity of a quantum channel.

One of the most useful general bounds on the Shannon capacity of a classical channel comes 
from~$\vartheta$, the Lov\'asz theta function \cite{lo}.
While, for classical channels, the complexity based bound is outperformed by the Lov\'asz number (see \cite[Theorem~11]{lo}), 
we will show that there exist quantum channels for which the quantum complexity bound on 
capacity we suggest is better than the bounds arising from the non-commutative analogue~$\thetatilde$ of the Lov\'asz number 
introduced in \cite{dsw}. 
In fact, we will show that there exist quantum channels $\Phi_k$ for $k\in \bN$ for which the ratio of the quantum Lov\'asz theta number $\thetatilde(\Phi_k)$
to the quantum complexity $\gamma(\Phi_k)$ introduced herein is arbitrarily large, while the upper bound $\gamma(\Phi_k)$ for the quantum Shannon zero-error capacity $\Theta(\Phi_k)$ is accurate to within a factor of two (see Corollary~\ref{cor:betterbounds}).

We
will see that the classical complexity of a graph $G$ is a familiar parameter which  
coincides with its {\it intersection number} (provided $G$ lacks isolated vertices).
For operator systems, the measure of quantum complexity we propose has not been previously studied.
We will characterise it in several different ways.
Since every graph $G$ gives rise to a canonical operator system $\cl S_G$, it can in addition be endowed with a 
{\it quantum complexity,} which can be strictly smaller than its classical counterpart, and can be equivalently characterised as a 
\emph{quantum intersection number} of~$G$. \goodbreak
\medskip

The paper is organised as follows. 
In Section~\ref{sec:graphs}, we begin by recalling the graph theoretic parameters needed in the sequel 
and show that our measure of the classical complexity of a graph coincides with its intersection number.  
In Section~\ref{sec:qint}, we turn to the quantum complexity of a graph, and show that it coincides with 
its minimum semi-definite rank (modulo any isolated vertices). 
In Section~\ref{sect:complexity}, we achieve a parallel development for operator systems, considering simultaneously 
a closely related notion of subcomplexity that coincides with the quantum chromatic number  
introduced in \cite{stahlke}.
We show that our operator system parameters are genuine extensions of the graph theoretic ones (Theorem~\ref{theorem:gen}) and explore similarities and differences between their behaviour on commutative and non-commutative graphs.
In Section~\ref{sec:applications-to-capacity}, we establish the bounds on capacities in terms of complexities (Theorem~\ref{th_boundbeta}) and show by example that these bounds can improve dramatically on the Lov\'asz $\vartheta$ bound. Finally, in Appendix~\ref{sec:order-proofs} we establish the partial ordering among various bounds on the quantum Shannon zero-error capacity, from this paper and elsewhere.

\medskip

In the sequel, we employ standard notation from linear algebra:
we denote by $M_{k,n}$ the space of all $k$ by $n$ matrices with complex entries, 
and set $M_n = M_{n,n}$. 
We let $\|X\|$ be the operator norm of a matrix $X\in M_{k,n}$, so that $\|X\|^2$ is the largest eigenvalue of $X^*X$.
We equip $M_n$ with the inner product given by $\langle X,Y\rangle = \tr(Y^*X)$, where $\tr(Z)$ is the 
trace of a matrix $Z\in M_n$. 
We write $I_n$ (or simply~$I$) for the identity matrix in $M_n$. 
The positive cone of $M_n$ 
(that is, the set of all positive semi-definite $n$ by $n$ matrices) will be denoted by $M_n^+$;
if~$\S\subseteq M_n$, we let $\S^+ = \S\cap M_n^+$.
We write
$\bb{R}^k_+$ for the cone of all vectors in $\bb{R}^k$ with non-negative entries, and
let $(e_i)_{i=1}^k$ be the standard basis of $\bb{C}^k$.
If $v,w\in \bb{C}^n$, we denote by $vw^*$ the rank one operator on $\bb{C}^n$ given by $(vw^*)(z) = \langle z,w\rangle v$, $z\in \bb{C}^n$.
The cardinality of a set $S$ will be denoted by $|S|$.

\section{Graph parameters}\label{sec:graphs}

In this section, we recall some graph theoretic parameters 
and point out their connection with Shannon's confusability graphs and channel capacities.
We start by establishing notation and terminology. Unless otherwise stated, all graphs in this paper will be simple graphs: undirected graphs without loops and at most one edge between any pair of vertices.
Let $n\in \bN$ and let $G$ be a graph with vertex
set~$[n] := \{1,\dots,n\}$.  For $i,j\in [n]$ we write $i\simeq j$ or
$i\simeq_G j$ to denote non-strict adjacency: either $i=j$, or $G$
contains the edge~$ij$. %
We denote by $G^c$ the \emph{complement} of the graph $G$; by definition, $G^c$ has vertex set $[n]$
and, for distinct $i,j\in [n]$, we have $i\simeq_{G^c} j$ if and only if $i\not\simeq_{G^c} j$.
For graphs $H,G$ with vertex set~$[n]$, we write $H\subseteq G$, and say that $H$ is a subgraph of $G$, if
every edge of $H$ is an edge of~$G$. %
An \emph{independent set} in $G$ is a subset of its vertices between which there are no edges of~$G$.
Let $k\in \bN$, and for an $n$-tuple $x = (x_1,\dots,x_n)$, where
each $x_i\in \bC^k$ is a non-zero vector, we define $G(x)$ to be the \emph{non-orthogonality graph} of~$x$, with vertex set~$[n]$ and adjacency relation given by
\[ i\simeq_{G(x)}j\iff \langle x_j,x_i\rangle\ne 0.\]

Let $G$ be a graph with vertex set $[n]$. 
Consider the following graph parameters:

\begin{itemize}
\item[(a)] the \emph{independence number} $\alpha(G)$ of $G$, given by 
  \[\alpha(G) = \max\{|S|\colon \text{$S$ is independent in~$G$}\};\]
\item[(b)] the \emph{quantum complexity}
  \[
    \hspace*{-.6cm}\gamma(G) = \min\{k\in \bN\colon \text{$G(x)=G$ for some $x\in (\bC^k\setminus\{0\})^n$}\}\]
  and the \emph{quantum subcomplexity}
  \begin{align*}
    \hspace*{-.6cm}\beta(G) &= \min\{k\in \bN\colon \text{$G(x)\subseteq G$ for some $x\in (\bC^k\setminus\{0\})^n$}\}\\
    &= \min\{\gamma(H)\colon \text{$H$ is a subgraph of~$G$}\};
  \end{align*}
\item[(c)] 
the \emph{intersection number} $\cc(G)$ of~$G$, given by 
\ifieee
  \begin{multline*} \cc(G) = \min\{k\in \bN\colon \text{$G(x)=G$ }\ifieee\\\fi\text{for some $x\in(\bR_+^k\setminus\{0\})^n$}\}.
  \end{multline*}
\else
\[\cc(G) = \min\{k\in \bN\colon \text{$G(x)=G$ }\text{for some $x\in(\bR_+^k\setminus\{0\})^n$}\}.\]
\fi
\end{itemize}

\begin{remark}\label{r_21}
  The graph parameters defined above are well known, and will be generalised to non-commutative graphs in Section~\ref{sect:complexity}.

\smallskip

\noindent {\bf (i)}  The independence number $\alpha(G)$ is standard in graph
    theory \cite{gr}. 

\smallskip

\noindent {\bf (ii)} Writing $G^c$ for the complement of~$G$, we have
    $\beta(G)=\xi(G^c)$ where $\xi$ is the orthogonal rank (see, for example,~\cite{ss}). The
    parameter $\gamma(G)$ is the minimum vector rank of~$G$ in the
    terminology of~\cite{jmn}, and is equal to $\msr(G)+|\iso(G)|$
    where $\msr(G)$ is the classical minimum semidefinite rank
    of~$G$~\cite{fh,hprs} and $\iso(G)$ is the set of isolated
    vertices of~$G$.
    \smallskip

\noindent {\bf (iii)} Let $\inter_{\rm st}(G)$ be the set-theoretic 
    intersection number of~$G$; thus, $\inter_{\rm st}(G)$ is the smallest positive integer $m$ 
    for which there exist non-empty sets $R_i \subseteq [m]$, $i = 1,\dots,n$, such that $i\simeq_G j$ if and only if $R_i\cap R_j \neq \emptyset$. (Note that usually in the literature one relaxes the assumption that the sets $R_i$ be 
    non-empty~\cite{mm}; however, it is more convenient for us to work with the definition above.)
   We claim that $\cc(G) = \inter_{\rm st}(G)$. 
   Indeed, first suppose that $\iso(G) = \emptyset$, and let $m = \inter_{\rm st}(G)$. Choose 
    $R_i\subseteq [m]$ for $i\in [n]$, so that
    $i\simeq_G j\iff R_i\cap R_{j}\ne \emptyset$. Note that, since $\iso(G)=\emptyset$, we have $R_i\ne\emptyset$ for every~$i$. Defining
    $x_i = \sum_{r\in R_i}e_r$ for $i\in [n]$, we have that
    $x_i\in \bR_+^k$, $i\in [n]$, and $G(x) = G$; hence $\cc(G)\leq k$. 
    Conversely, suppose that $x = (x_1, \dots, x_n)$ is a tuple of non-zero vectors in $\bb{R}^k_+$ such that 
    $G(x) = G$. Let $R_i = \{l\in [k] : \langle x_i,e_l\rangle \neq 0\}$. The non-negativity of the entries of $x_i$, $i = 1,\dots,n$, 
    implies that $i\simeq_G j$ if and only if $R_i\cap R_j \neq \emptyset$; thus, $\inter_{\rm st}(G)\leq \cc(G)$ and so 
    $\inter_{\rm st}(G) = \cc(G)$. 

\smallskip

  It is straightforward from the definitions that
\begin{equation}\label{eq_abgi}
\alpha(G)\leq \beta(G)\leq \gamma(G)\leq \cc(G)
\end{equation}
for every
  graph~$G$. We note that these inequalities will be generalised to arbitrary operator systems in $M_n$ 
  in Theorem~\ref{thm:inncc} below.
\end{remark}

We now review some of Shannon's ideas \cite{shannon}.
Suppose that we have a finite set $X$, which we view as an alphabet
that we wish to send through a noisy channel $\cl N$ in order to obtain symbols from another alphabet, say $Y$. 
We let $p(y|x)$ denote the probability that, if we started with the symbol $x \in X$, then after this process, 
the symbol $y \in Y$ is received. We require that every $x\in X$ is transformed into some $y\in Y$, that is,
$\sum_{y \in Y} p(y|x) = 1$, for all $x \in X$. 
The column-stochastic matrix $\cl N = (p(y|x))$, indexed by $Y \times X$, 
is often referred to as the {\it noise operator} of the channel. 
We will write $\cl N : X \to Y$ to indicate the matrix $(p(y|x))$, 
and refer to such matrices as {\it (classical) channels}. The {\it confusability graph} of $\cl N$ is the graph $G_{\cl N}$ 
with vertex set $X$ for which, given two distinct $x,x'\in X$, 
the pair $xx^{\prime}$ is an edge if and only if there exists $y \in Y$ such that $p(y|x)p(y|x^{\prime}) > 0$. 
Equivalently, $x\simeq_{G_{\cl N}} x'$ if and only if there exists $y\in Y$ such that the symbols $x$ and $x'$ 
can be transformed into the same $y$ via $\cl N$ and hence confused. 

The {\it one-shot zero-error capacity of $\cl N$}, denoted $\alpha(\cl N)$, 
is defined to be the cardinality of the largest subset $X_1$ of $X$ such that, whenever an element of $X_1$ is sent via $\cl N$, 
no matter which element of $Y$ is received, the receiver can determine with certainly 
the input element from $X_1$. It is straightforward that $\alpha(\cl N) =\alpha(G_{\cl N})$.

\begin{definition}\label{d_plexi}
Let $G$ be a graph with vertex set $[n]$. The \emph{complexity $\plex(G)$ of $G$} is the 
minimal cardinality of a set $Y$ such that $G = G_{\cl N}$ for some channel $\cl N : [n] \to Y$.  

If $\cl N:X \to Y$ is a channel, we set $\plex(\cl N) = \plex(G_{\cl N})$ and call it 
this parameter the \emph{complexity of $\cl N$}. 
\end{definition}  

\begin{proposition}\label{p_plexeqcc} 
Let $G$ be a graph with vertex set $[n]$. Then $\plex(G) = \cc(G)$. 
In other words, if $\cl N : X\to Y$ is  a channel then $\plex(\cl N) = \inter(G_{\cl N})$.
\end{proposition}

\begin{proof} 
Let $\cl N:[n] \to Y$ be a channel so that $G= G_{\cl N}$. For $1 \le i \le n$ set $R_i = \{ y \in Y: p(y|i) \ne 0 \}$.  Then $R_i$ is non-empty for each $i$  and $i \simeq_G j \iff R_i \cap R_j \ne \emptyset$.  This shows that $\cc(G)$ is a lower bound for the complexity of $G$. 

Conversely, suppose that $R_1,\dots,R_n$ are non-empty subsets of $[k]$ such that 
$i\simeq_G j$ if and only if $R_i\cap R_j\neq \emptyset$
Set
\[ p(y|i) =\begin{cases} 1/|R_i|, & y \in R_i \\ 0, & y \notin R_i \end{cases};\]
then $\cl N= (p(y|i))$ is a channel from $[n]$ to $[k]$ with $G_{\cl N} = G$. This shows that   
$\cc(G)$ is an upper bound for the complexity of $G$.
\end{proof}

\begin{remark}\label{r_123}
{\bf (i)}
We will discuss later (Remark~\ref{rk:classical-quantum}) the natural way to view a classical channel $\cl N$ as a quantum channel; 
we will see that the quantum complexity of $\cl N$, studied in Section~\ref{sect:complexity}, coincides with $\gamma(G_{\cl N})$.

\smallskip

\noindent  {\bf (ii)}  It is well-known that $\beta(G)\leq \chi(G^c)$ for any graph~$G$
(here $\chi(H)$ denotes the chromatic number of a graph $H$ \cite{gr}), so
  it is natural to ask if $\chi(G^c)$ fits into chain of
  inequalities (\ref{eq_abgi}).  
  In fact, it does not: one can check using a computer program
  that $\chi(G^c)\leq \gamma(G)$ for all graphs on $7$ or fewer
  vertices, but this inequality fails in general, for example if~$x$ is
  a Kochen-Specker set and $G=G(x)$ (see~\cite[Section~1.2]{hpswm}).
  
  \smallskip

\noindent {\bf (iii)}  For each $\pi\in\{\alpha,\beta,\gamma,\cc\}$, we have that $\pi(G)=1$ if and
  only if~$G$ is a complete graph.
  Indeed, if~$G$ is a complete graph, then $\pi(G)\leq \inter(G)=1$, so
  $\pi(G)=1$; and if $G$ is not a complete graph, then $G^c$ contains
  at least one edge, so $\pi(G)\ge \alpha(G)>1$.
\end{remark}

\section{The quantum intersection number}\label{sec:qint}

In this section we show that the graph parameter~$\gamma$, discussed in Section~\ref{sec:graphs}, has a
reformulation in terms of projective colourings of the graph $G$, which leads to a parameter that we call the
quantum intersection number of $G$. This will allow a key step in the
proof of Theorem~\ref{theorem:gen}, where we show that $\gamma$ has a
natural operator system generalisation.

Fix~$n\in \bN$. Given $k\in\bN$, we write $\P(k)$ for the set of 
$n$-tuples $P=(P_1,\dots,P_n)$ where each $P_i$ is a non-zero
projection in $M_k$.
Let $\P_c(k)$ denote the subset of $\P(k)$ consisting of the elements $P = (P_1,\dots,P_n)$ with commuting entries:
$P_iP_j = P_jP_i$ for all~$i,j$.
To any~$P\in \P(k)$ we associate the \emph{non-orthogonality graph} $G(P)$
with vertex set~$[n]$ and edges defined by the relation 
\[ i\simeq_{G(P)}j\iff P_iP_j\ne0.\]
We define the \emph{quantum intersection number} $\qinter(G)$ of a graph $G$ with vertex set $[n]$ by letting 
\[\qinter(G)=\min\left\{k\in \bN\colon \text{$G(P)=G$ for some $P\in \P(k)$}\right\}.\]
The next proposition explains the choice of terminology.

\begin{proposition}\label{p_intintq}
Let $G$ be a graph with vertex set $[n]$. Then 
\begin{equation}\label{eq_ccG}
\cc(G)=\min\left\{k\in\bN\colon\text{$G(P)=G$ for some  $P\in \P_c(k)$}\right\}.
\end{equation}
\end{proposition}

\begin{proof}
Let~$l$ be the minimum on the right hand side of (\ref{eq_ccG}), and suppose that 
$G(x)=G$ for some $n$-tuple $x$ of non-zero vectors in $\bR_+^k$. 
Letting $P_i$ be the orthogonal projection onto the linear span of 
$\{e_r\colon \langle x_i,e_r\rangle\ne 0\}$ yields a tuple $P=(P_1,\dots,P_n)\in \P_c(k)$ 
with $G(P) = G$; thus, $l \leq \cc(G)$. 

Conversely, suppose that $P = (P_1,\dots,P_n)\in \P_c(l)$ is such that $G(P)=G$.
Simultaneously diagonalising the $P_i$'s with respect to a basis $\{b_r\colon r\in [l]\}$ and defining  
$R_i = \{r\in [l]\colon P_ib_r \ne 0\}$, $i\in [n]$, we see that $i\simeq_G j$ if and only if $R_i\cap R_j \neq \emptyset$,
and it follows that $\cc(G)\leq l$.
\end{proof}

Let $t = (t_1,\dots,t_n)\in \bN^n$, and write $|t|=\sum_{i=1}^n t_i$. 
Extending ideas from~\cite{hprs}, let
\[ \P(k,t) = \left\{(P_1,\dots,P_n)\in \P(k)\colon \rank P_i=t_i,\ i\in [n]\right\}\]
and define
\[ \mpsd[t](G)=\min \left\{k\in \bN\colon \text{$G(P)=G$ for some $P\in \P(k,t)$}\right\}.\]
Consider each element $A\in
M_{|t|}$ as a block matrix $A=[A_{i,j}]_{i,j\in [n]}$ where
$A_{i,j}\in M_{t_i,t_j}$. 
We define
\ifieee
\begin{multline*}
 \F^+_t(G)= \left\{A\in M_{|t|}^+\colon A_{i,j}\ne 0\Leftrightarrow i\simeq j,\right.
\\
\left.\vphantom{M_{|t|}^+}\text{and $\rank(A_{i,i})=t_i$ for each~$i\in [n]$}\right\}
\end{multline*}
\else
\[ \F^+_t(G)= \left\{A\in M_{|t|}^+\colon A_{i,j}\ne 0\Leftrightarrow i\simeq j,\text{and $\rank(A_{i,i})=t_i$ for each~$i\in [n]$}\right\}\]
\fi
and
\[ \H^+_t(G) = \left\{A\in\F^+_t(G)\colon A_{i,i}=I_{t_i} \text{ for
  each~$i\in [n]$}\right\}.\] 
  We write $\F^+(G)=\F^+_{\mathbbm 1}(G)$ and
$\H^+(G)=\H^+_{\mathbbm 1}(G)$, where $\mathbbm 1=(1,1,\dots,1)$.
  Note that in~\cite{hprs}, $\mpsd[t](G)$ and $\H^+_t(G)$ were defined in the special case where 
  $t=(r,r,\dots,r)$ for some~$r\in \bN$.  

\begin{proposition}\label{prop:F+H+}
Let $G$ be a graph with vertex set $[n]$ and let $t\in \bN^n$. Then
\ifieee
\begin{align*}
  \mpsd[t](G)&=\min\left\{ \rank A : A \in\F^+_t(G) \right\}
        \\&= \min \left\{ \rank B : B \in \H^+_t(G) \right\}.
\end{align*}\else
\[\mpsd[t](G)=\min\left\{ \rank A : A \in\F^+_t(G) \right\}
        = \min \left\{ \rank B : B \in \H^+_t(G) \right\}.\]
\fi
\end{proposition}
\begin{proof}
The proof is an adaptation of~\cite[Theorem~3.10]{hprs}. 
Suppose first that $A \in \F^+_t(G)$ and $\rank A = k$. Then
  there exists a matrix $X \in M_{k , |t|}$ such that $A = X^* X$.
Write $X = [ X_1 \ X_2\ \cdots\ X_n]$, where $X_i\in M_{k,t_i}$, $i = 1,\dots,n$.  
We have 
$$\rank
  X_i=\rank(X_i^*X_i)=\rank(A_{ii})=t_i, \ \ \ i\in [n].$$
Writing $P=(P_1,\dots,P_n)$ where $P_i\in M_k$ is the orthogonal projection onto the range of~$X_i$, we have $P\in \P(k,t)$.
  Additionally, 
  $$P_iP_j\ne0\iff X_i^*X_j \ne  0 \iff A_{i,j} \ne  0  \iff i\simeq j,$$ 
  so $G(P)=G$. 
  Hence \[\mpsd[t](G) \leq \min \left\{ \rank A : A \in \F^+_t(G) \right\}.\]

  Since $\H^+_t(G)\subseteq\F^+_t(G)$, the inequality \[\min\left\{
      \rank A : A \in \F^+_t(G) \right\} \leq \min\left\{ \rank B : B
      \in \H^+_t(G) \right\}\] holds trivially.
  
  Now suppose that $P=(P_1,\dots,P_n)\in \P(k,t)$ with $G(P)=G$, and
  for each $i \in [n]$ let $X_i \in M_{k , t_i}$ be a matrix whose
  columns form an orthonormal basis
  for
  the range of~$P_i$. Define
  $X = [ X_1\ X_2\ \cdots\ X_n]\in M_{k , |t|}$ and let
  $B = X^* X \in M_{|t|}^+$, so that $\rank B = \rank X \leq k$. Note
  that the $t_i\times t_j$ block $B_{i,j}$ coincides with
  $X_i^* X_j$. Since
  $i\simeq j\iff P_i P_j\ne0\iff X_i^* X_j \ne 0$, we have
  $B_{ij} \ne 0\iff i\simeq j$. Moreover, the condition on the columns of 
$X_i$ implies that $B_{i,i} = I_{t_i}$ for each~$i$. So
  $B \in \H^+_t(G)$, hence
  \[\min \left\{ \rank B : B \in \H^+_t(G) \right\} \leq
    \mpsd[t](G).\qedhere\]
\end{proof}

\begin{theorem}\label{thm:qinter=mpsd}
  For any graph~$G$, we have $\qinter(G)=\gamma(G)$.
\end{theorem}
\begin{proof}
  Directly from the definitions, we have
 \[ \qinter(G)=\min_{t\in \bN^n} \mpsd[t](G)\leq \mpsd[\mathbbm 1](G)=\gamma(G).\]
  Let~$t\in \bN^n$. We claim that if $t_1\ge2$ and
  $s=(t_1-1,t_2,t_3,\dots,t_n)$, then $\mpsd[s](G)\leq
  \mpsd[t](G)$. By symmetry and induction, this yields
  $\gamma(G)=\mpsd[\mathbbm 1](G)\leq \mpsd[t](G)$ for
  any~$t\in \bN^n$, hence $\qinter(G) = \gamma(G)$.

  To establish the claim, suppose that~$t_1\ge2$, let $k=\mpsd[t](G)$, and use Proposition~\ref{prop:F+H+} to choose $B\in
  \H_t^+(G)$ with $\rank B=k$. We may write $B=X^*X$ where $X\in
  M_{k,|t|}$. Write $X=[X_1\ X_2\ \cdots\ X_n]$ where $X_i\in
  M_{k,t_i}$, let $Y\in M_{k,|t|-1}$ be $X$ with the first
  column deleted and let $Z_1\in M_{k,t_1-1}$ be the matrix with
  every column equal to the first column of~$X$. Let $Z=[Z_1\ 0\
  \dots\ 0]\in M_{k,|t|-1}$.

  Let $Y_1\in M_{k,t_1-1}$ consist of the first $(t_1-1)$ columns
  of~$Y$. Note that $X^*Y_1$ contains $X_1^*Y_1$ as a submatrix, which
  is equal to $I_{t_1}$ with the first column removed. In particular,
  $X^*Y_1\ne 0$. Similarly, the first column of $X_1^*Z_1$ is the first
  column of~$I_{t_1}$, so $X^*Z_1\ne 0$. Define
  \begin{gather*}
    a=\min\{|w|\colon \text{$w\ne 0$, $w$ is an entry of $X^*Y_1$}\}\\
    b=\max\{|w|\colon \text{$w$ is an entry of $X^*Z_1$}\}
  \end{gather*}
  and let $\epsilon\in (0,\tfrac12 ab^{-1})$. Define
  \[W_\epsilon=Y+\epsilon Z\in M_{k,|t|-1}\ifieee\]and\[\else\qand\fi
  A=A_\epsilon = W_\epsilon^* W_\epsilon \in M_{|t|-1}^+.\] Note that $\rank A=\rank W_\epsilon \leq
  k$. Let us write $W_\epsilon = [W_1\ W_2\ \cdots\ W_n]$ and $A=[A_{i,j}]_{i,j\in
    [n]}$, with block sizes given by $s=(t_1-1,t_2,\dots,t_n)$. Note that 
  if $2\leq i,j\leq n$, then $W_i=X_i$ and $W_j=X_j$, hence
  $A_{i,j}=W_i^*W_j=X_i^*X_j=B_{i,j}\ne0\iff i\simeq j$. Let $i\in
  [n]$. We have
  \begin{equation}\label{eq:Ai1}
    A_{i,1}=X_i^*(Y_1+\epsilon Z_1)=X_i^*Y_1+\epsilon X_i^* Z_1.
  \end{equation}
  If $i\not\simeq 1$, then $X_i^*Y_1$ is a submatrix of $X_i^*X_1=0$,
  and $X_i^*Z_1$ is the matrix with every column equal to the first
  column of $X_i^*X_1=0$. Hence if $i\not\simeq1$, then
  $A_{i,1} = 0 = A_{1,i}$. Now suppose that $i\simeq 1$. Since $X_i^*Y_1$
  and $X_i^*Z_1$ are submatrices of~$X^*Y_1$ and $X^*Z_1$,
  respectively, by~(\ref{eq:Ai1}) and our choice of~$\epsilon$, if
  $X_i^*Y_1$ has any non-zero entry, then the corresponding entry of
  $A_{i,1}$ is also non-zero. On the other hand, if $X_i^*Y_1=0$, then
  since $i\simeq1$ yields $X_i^*X_1\ne0$, we must have $X_i^*Z_1\ne0$,
  hence $A_{i,1}\ne0$; since $A=A^*$, we also have $A_{1,i}\ne0$.  This
  shows that for any $\epsilon>0$, the matrix~$A=A_\epsilon$ satisfies
  \[A_{i,j}\ne0\iff i\simeq j,\quad\text{for any~$i,j\in [n]$}.\]
  Since~$W_\epsilon\to Y$ as~$\epsilon\to 0$, we see that $A_\epsilon$
  converges to the matrix~$B$ with the first row and column removed;
  in particular, the top left $(t_1-1)\times (t_1-1)$ block
  of~$A_\epsilon$ converges to~$I_{t_1-1}$. Hence by choosing
  $\epsilon\in (0,\tfrac12ab^{-1})$ sufficiently small, we may ensure
  that $A_{1,1}$ has rank~$t_1-1$, hence $A\in \F_s^+(G)$. By Proposition~\ref{prop:F+H+}, we have 
  $\mpsd[s](G) \leq \rank A\leq k = \mpsd[t](G)$.
 \end{proof}

\section{Quantum channels and operator system parameters}\label{sect:complexity}

We recall that 
an {\it operator subsystem of $M_n$} is a subspace $\cl S \subseteq M_n$ 
such that $I \in \cl S$ and $X \in \cl S \implies X^* \in \cl S$. In this paper we will sometimes refer to such a self-adjoint unital subspace $\S\subseteq M_n$ simply as an \emph{operator system}; we refer the reader to~\cite{paulsen_2002} for the general theory of operator systems and completely bounded maps.
A linear map $\Phi:M_n \to M_k$ is called a {\it quantum channel} 
if it is completely positive and trace-preserving. By theorems of Choi and Kraus, $\Phi$ is a quantum channel if and only if there exists $m\in \bN$ and matrices 
$A_1,\dots,A_m\in M_{k,n}$,  satisfying $\sum_{i=1}^m A_i^*A_i = I_{n}$, so that
\[\Phi(X) = \sum_{i=1}^m A_i X A_i^*, \  \ \ X\in M_{n}.\]

This realisation of $\Phi$ is called a {\it Choi-Kraus representation} 
and the matrices $A_i$ are called its {\it Kraus operators.}  
The Choi-Kraus representation is far from unique, but it was 
shown in \cite{dsw} that the subspace of $M_n$ spanned by the set
$\{ A_i^*A_j: 1 \le i,j \le m \}$ is independent of it.   Consequently, \cite{dsw} set
\[ \cl S_{\Phi} := \spn \left\{ A_i^*A_j: 1 \le i,j \le m \right\}, \]
where $\Phi(X) = \sum_{i=1}^m A_iXA_i^*$ 
is any Choi-Kraus representation of $\Phi$. 
This space, easily seen to be an operator system, is called the {\it non-commutative confusability graph of $\Phi$.}

Regarding operator subsystems of $M_n$ as non-commutative confusability graphs, 
we wish to define operator system analogues of the graph parameters considered in Section \ref{sec:graphs}.
Just as every graph is the confusability graph of some classical channel,
\cite{dsw} showed that the map $\Phi\mapsto \S_\Phi$ 
from quantum channels with domain $M_n$ to operator subsystems of $M_n$, is surjective.
We will need the following estimate on the dimension of the target Hilbert space.

\begin{proposition}%
\label{p_Phi}
  Let $n\in \bb{N}$. If $\S\subseteq M_n$ is an operator system, then
  there exists $k\in \bN$ and a quantum channel
  $\Phi\colon M_n\to M_k$ such that $\S = \S_{\Phi}$. In fact, if
  $m\in \bN$ is such that $\binom{m}{2}\ge \dim \S-1$, then we can take $k=mn$.
\end{proposition}
\begin{proof}
Let~$d=\dim \S$ and let $I_n,S_1,\dots,S_{d-1}$ be
  a basis of~$\S$. Suppose that $\binom{m}{2}\ge \dim \S-1$. Then we can form a hermitian
  $m\times m$ block matrix $H=[H_{ij}]_{i,j\in [m]}\in M_m(M_n)$ so that
  \[\text{$H_{ii}=I_n$ for $i\in [m]$}\ifieee\]and\[\else \text{ and }\fi
    \{S_1,\dots,S_{d-1}\}=\{H_{ij}\colon 1\leq i<j\leq m\}.\] For
  sufficiently small $\epsilon>0$, the matrix $X=\frac1{m}(I_{mn}+\epsilon H)$ is
  positive semi-definite, hence $X=C^*C$ for some $C\in M_m(M_n)$, and
  the block entries of~$X$ span~$\S$. The $mn\times n$ block columns
  of~$C$ are then Kraus operators for a quantum
  channel~$\Phi\colon M_n\to M_{mn}$ for which $\S_\Phi$ is spanned
  by the entries of~$X$, so $\S_\Phi=\S$.
\end{proof}

We now define parameters of operator systems which, as we will shortly see, generalise the graph parameters above.
Let $\S\subseteq M_n$ be an operator system. As usual, we write 
$$\S^\perp=\{A\in M_n\colon \tr(A^*S)=0\text{ for all $S\in \S$}\}.$$

\begin{itemize}  
\item[(a)] 
Let $\cl S\subseteq M_n$ be an operator system.
Recall \cite{dsw} that an $\S$-independent set of size $m$ is an $m$-tuple $x=(x_1,\dots,x_m)$ 
with each~$x_i$ a non-zero vector in $\bC^n$, so that $x_px_q^*\in \S^\perp$ whenever $p,q\in [m]$ with $p\ne q$. 
The independence number~$\alpha(\S)$ is then defined by letting
\[\hspace*{-.5cm} \alpha(\S)=\max\{m\in \bN\colon \text{\ifieee$\exists$\else there is\fi~an $\S$-independent set of size~$m$}\}.\]
\item[(b)] 
We define the 
\emph{quantum complexity} $\gamma(\S)$ by letting
\ifieee\begin{multline*}
  \gamma(\S) = \min\{k\in \bN \colon \S_\Phi=\S \\\text{ for some quantum channel $\Phi \colon M_n\to M_k$}\}
\end{multline*}
\else
\[
  \gamma(\S) = \min\{k\in \bN \colon \S_\Phi=\S \text{ for some quantum channel $\Phi \colon M_n\to M_k$}\}
\]
\fi
and the
\emph{quantum subcomplexity} $\beta(\S)$ by letting
\ifieee
\begin{align*}
  \beta(\S) &= \min\{k\in \bN \colon \S_\Phi\subseteq\S\\&\hphantom{={}}\text{ for some quantum channel $\Phi \colon M_n\to M_k$}\} 
    \\&= \min \{ \gamma(\cl T): \cl T \subseteq \cl S \text{ is an operator subsystem}\}.
  \end{align*}
\else
\begin{align*}
  \beta(\S) &= \min\{k\in \bN \colon \S_\Phi\subseteq\S\text{ for some quantum channel $\Phi \colon M_n\to M_k$}\} 
    \\&= \min \{ \gamma(\cl T): \cl T \subseteq \cl S \text{ is an operator subsystem}\}.
  \end{align*}
\fi
  
\item[(c)] 
A quantum channel~$\Phi$ which has a set of Kraus operators each of which is of the form $AD$ for some entrywise non-negative matrix $A$ and an invertible diagonal matrix~$D$ will be said to be a \emph{non-cancelling}. We define
  \begin{multline*} \cc(\S) = 
  \inf\{k\in \bN\colon \text{$\S_\Phi=\S$}\ifieee\text{ for some}\\\else\\\text{for some }\fi\text{non-cancelling quantum channel $\Phi \colon M_n\to M_k$}\}.
  \end{multline*}
\end{itemize}

\begin{corollary}\label{c_bong} 
Let $\cl S\subseteq M_n$ be an operator system. Then $\gamma(\cl S)\leq 2n^2$.
\end{corollary}
\begin{proof}
Since $\dim\cl S\leq n^2$, we can take $m = 2n$ in Proposition \ref{p_Phi}.
\end{proof}

We will refer to $\gamma(\Phi)$ as the {\it quantum complexity of $\Phi$} and $\beta(\Phi)$ 
as the {\it quantum subcomplexity of $\Phi$.}
Given a channel $\Phi$, we set $\pi(\Phi) = \pi(\cl S_{\Phi})$ for $\pi \in \{ \beta, \gamma, \inter\}$. 

\begin{remark}\label{remark:pi=1}
{\bf (i) }
A set of quantum states can be perfectly distinguished by a measurement system if and only if they are orthogonal.  
Consequently, \cite{dsw} defined the {\it one-shot zero-error capacity $\alpha(\Phi)$} of a quantum channel 
$\Phi:M_n \to M_k$ to be the maximum cardinality of a set $\{ v_1, ..., v_p \} \subseteq \bb C^n$  orthogonal unit vectors, 
such that 
\[ \tr(\Phi(v_iv_i^*) \Phi(v_jv_j^*)) =0, \ \ \ i \ne j.\]
It was shown in~\cite{dsw} (see also~\cite{p_notes}) that $\alpha(\Phi) = \alpha(\cl S_{\Phi})$. 

\smallskip

\noindent {\bf (ii) } 
Let $\cl S$ be an operator system. 
The \emph{quantum chromatic number} $\chi_q(\cl S^{\perp})$ of 
the orthogonal complement $\cl S^{\perp}$ of $\cl S$ was introduced by D. Stahlke in \cite{stahlke}.
It is straightforward that $\beta(\cl S) = \chi_q(\cl S^{\perp})$.

\smallskip

\noindent {\bf (iii) }  For an operator system~$\S\subseteq M_n$ and
  $\pi\in\{\beta,\gamma,\cc\}$, we have $\pi(\S)=1 \iff \S=M_n$.
  Indeed, the trace $\tr\colon M_n\to \bC$ is a
  non-cancelling quantum channel since it has the entry-wise non-negative Kraus operators
  $e_1^*,\dots,e_n^*$
  (where $e_i^*$ is the functional corresponding to the vector $e_i$), 
  so  $1\leq \pi(M_n)\leq \cc(M_n)=1$ and hence
  $\pi(M_n)=1$.
  Conversely, $\pi(\S)=1$ implies that $\beta(\S)=1$; the trace
  is the only scalar-valued quantum channel on~$M_n$, so
  $M_n = \S_{\tr}\subseteq \S\subseteq M_n$, that is, $\S = M_n$.

  Note that, in contrast with Remark~\ref{r_123} (iii), it is not true
  that $M_n$ is the only operator system $\cl S$ with $\alpha(\S)=1$; see Proposition~\ref{prop:n=2,gamma=3}.
  
  \smallskip

\noindent  {\bf (iv) } We claim that \[\alpha(\bC I_n)=\beta(\bC I_n)=\gamma(\bb{C}I_n)=\inter(\bC I_n) = n.\] Indeed, one sees immediately that $\alpha(\bC I_n)\geq n$ by considering the $\bC I_n$-independent set $(e_1,\dots,e_n)$, and since the identity channel $M_n\to M_n$ is non-cancelling, we have that $\inter(\bb{C}I_n) \leq n$, so an appeal to Theorem~\ref{thm:inncc} below establishes the claim.
\smallskip

\noindent {\bf (v) } Let $\cl S\subseteq M_n$. Using (iv), we have
\[\beta(\cl S) = \min\{\gamma(\cl T) : \cl T\subseteq \cl S\} \leq \gamma(\bb{C}I_n) = n.\]
On the other hand, $\gamma(\S)$ may exceed $n$, even for $n=2$ (see Proposition~\ref{prop:n=2,gamma=3}).
\smallskip

\noindent {\bf (vi) }
There exist operator systems $\S\subseteq M_n$ with $\inter(\S)=\infty$, so the infimum in the definition of $\inter(\S)$ cannot be replaced with a minimum. For example, it is not difficult to see that this is the case for the two-dimensional operator system $\S\subseteq M_4$ spanned by the identity and $H=\begin{sbmatrix}0&X\\X&0\end{sbmatrix}$ where $X=\begin{sbmatrix}1&1\\1&-1\end{sbmatrix}$.
\end{remark}

\begin{theorem}\label{thm:inncc} 
Let $\cl S\subseteq M_n$ be an operator system. Then
$$\alpha(\cl S)\leq \beta(\cl S)\leq \gamma(\cl S)\leq \cc(\S).$$
\end{theorem}

\begin{proof}
Suppose that $\Phi : M_n\to M_k$ is a quantum channel
with $\cl S_{\Phi}\subseteq \cl S$; let $A_i \in M_{k,n}$, $i = 1,\dots,d$, be its Kraus operators.
Let $(x_p)_{p=1}^m$ be an $\cl S$-independent set of size~$m$.
For each $p\in[m]$, let $E_p$ be the projection in $M_k$ onto the span of $\{A_ix_p : i = 1,\dots,d\}$.
Since
$$\sum_{i=1}^d \|A_ix_p\|^2 = \left\langle \left(\sum_{i=1}^d A_i^* A_i\right) x_p,x_p\right\rangle = 1,$$
we have that $E_p\neq 0$ for each $p$.
On the other hand, since $A_j^* A_i\in \cl S$ for all $i,j = 1,\dots,d$ and  $(x_p)_{p=1}^m$
is $\S$-independent, we have that
$$\langle A_ix_p, A_jx_q\rangle = \langle A_j^*A_ix_p,x_q\rangle = 0, \ \ \  p\neq q, \ i,j = 1,\dots,d.$$
Thus, $E_1,\dots,E_m$ are pairwise orthogonal projections in~$M_k$; it follows that 
$m\leq k$ and hence $\alpha(\cl S)\leq \beta(\cl S)$.

The inequalities $\beta(\cl S)\leq \gamma(\cl S)\leq \cc(\S)$ hold trivially.
\end{proof}

In the next proposition, we collect some properties of the operator system parameters introduced above.

\begin{proposition}\label{prop:elem}
  Let~$\S\subseteq M_n$ and $\S_i\subseteq M_{n_i}$, $i=1,2$ be operator systems.
  \begin{itemize}
  \item[(i)] If $\pi\in\{\alpha,\beta,\gamma\}$ and~$U\in M_n$ is unitary, then $\pi(U^*S U)=\pi(\S)$;
  \item[(ii)] If $\pi\in\{\alpha,\beta,\gamma\}$ and~$P\in M_n$ is a projection of rank~$r$, then, viewing $P\S
    P$ as an operator subsystem of~$M_r$, we have $\pi(P\S P)\leq \pi(\S)$;
  \item[(iii)] If $\pi\in\{\beta,\gamma\}$, $n=n_1n_2$ and $\S=\S_1\otimes \S_2$, 
  then \[\max \{ \pi(\S_1), \pi(\S_2) \} \le \pi(\S)\leq\pi(\S_1)\pi(\S_2);\] 
  \item[(iv)] If $\pi\in \{\beta,\gamma,\inter\}$, $n=n_1=n_2$ and $\S=\spn(\S_1\cup \S_2)$, 
  then 
  $$\pi(\S)\leq \pi(\S_1)+\pi(\S_2).$$
  \item[(v)] If $\pi\in \{\beta,\gamma\}$, then $\pi(\S_1\oplus \S_2) = \pi(\S_1)+\pi(\S_2)$.
  \end{itemize}
\end{proposition}

\begin{proof} The proofs for $\pi=\alpha$ are easy and are left to the reader. We give the proofs for $\pi=\gamma$; the other proofs follow identical patterns.

(i) If~$\{A_p\}_{p=1}^m$ are Kraus operators in $M_{k,n}$ for which
    $\spn\{A_p^*A_q\}_{p,q=1}^m =\S$, then $\{A_pU\}_{p=1}^m $ are Kraus
    operators in $M_{k,n}$ such that
\ifieee
\begin{align*}
  \spn\{(A_pU)^*A_qU\}_{p,q=1}^m &= U^*\spn\{A_p^*A_q\}_{p,q=1}^m U \\&= U^*\S U,
\end{align*}
\else
\begin{align*}
  \spn\{(A_pU)^*A_qU\}_{p,q=1}^m &= U^*\spn\{A_p^*A_q\}_{p,q=1}^m U = U^*\S U,
\end{align*}
\fi
so
    $\gamma(U^*\S U)\leq \gamma(\S)$; the reverse inequality
    follows by symmetry.
    
(ii) If~$\{A_p\}_{p=1}^m$ are Kraus operators in $M_{k,n}$ for which
    $\spn\{A_p^*A_q\}_{p,q=1}^m =\S$, then after identifying the range of~$P$
    with $\bC^r$, we see that $\{A_pP\}_{p=1}^m$ are Kraus
    operators in $M_{k,r}$ with 
    $$\spn\{(A_pP)^*A_qP\}_{p=1}^m = P\spn\{A_p^*A_q\}_{p,q=1}^m P = P\S P,$$ so
    $\gamma(P\S P)\leq \gamma(\S)$.
    
(iii) Suppose that
    $\{A_{p,i}\colon p\in [m_i]\}\subseteq M_{k_i,n_i}$ is a family of Kraus operators for $i=1,2$, so
    that $\S_i = \spn\{A_{p,i}^*A_{q,i}\colon p,q\in [m_i]\}$, $i = 1,2$. 
Set $B_{p,r} := A_{p,1}\otimes A_{r,2}$; then $\{B_{p,r} \colon p\in [m_1],\,r\in [m_2]\}$ is a family of Kraus operators in 
$M_{k_1k_2,n}$ with $\S = \spn \{B_{p,r}^*B_{q,s}\colon p,q\in [m_1],\,r,s\in [m_2]\}$.
It follows that $\gamma(\cl S) \leq \gamma(\cl S_1) \gamma(\cl S_2)$. 

If we set $P_1 = I_{n_1} \otimes Q$ where $Q \in M_{n_2}$ is a rank one projection, then we have that $P_1(\S_1 \otimes \S_2) P_1 = \S_1 \otimes \bb C \equiv \S_1$. Hence by (ii), $\pi(\S_1) \le \pi(\S_1 \otimes \S_2)$ and the lower bound follows.

(iv) Let~$\{A_{p,i}\colon p\in [m_i]\}$ be a family of Kraus operators in $M_{k_i,n}$ with
    $\S_i = \spn\left\{A_{p,i}^*A_{q,i}\colon p,q\in [m_i]\right\}$, $i=1,2$.
Set $B_{p,1}=\left[\begin{smallmatrix}A_{p,1}\\0\end{smallmatrix}\right]$ and $B_{p,2}=\left[\begin{smallmatrix}0\\A_{p,2}\end{smallmatrix}\right]$,
viewed as elements of $M_{k_1+k_2,n}$. 
Then $\left\{\tfrac1{\sqrt2}B_{p,i}\colon i=1,2,\,p\in [m_i]\right\}$ is a family of Kraus operators with
$$\spn\{B_{p,i}^*B_{q,j}\}_{p,q,i,j}=\spn\{B_{p,i}^*B_{q,i}\}_{p,q,i}=\S,$$
so $\gamma(\S)\leq k_1+k_2$. Hence, $\gamma(\S)\leq \gamma(\S_1)+\gamma(\S_2)$.

(v) If $\S=\S_1\oplus \S_2$, then $n=n_1+n_2$. Suppose that $\gamma(\S)=k$, so that there exists a family  
$\{C_p=[A_p\;B_p]\colon p\in [m]\}$ of $k\times n$-Kraus operators, 
where $A_p\in M_{k,n_1}$ and $B_p\in M_{k,n_2}$, with $\S = \spn\{C_p^*C_q\}_{p,q=1}^m$. 
Since $C_p^*C_q=
\begin{sbmatrix}
  A_p^*A_q&A_p^*B_q\\B_p^*A_q&B_p^*B_q
\end{sbmatrix}\in \S_1\oplus\S_2$, we have $A_p^*B_q=0$, hence the
ranges of $A_p$ and $B_q$ are orthogonal for every~$p,q$. The
projections $P_1$ and $P_2$ onto the linear span of the ranges of
$A_1,\dots,A_m$ and $B_1,\dots,B_m$ are therefore orthogonal, so if
$k_1=\rank P_1$ and $k_2=k-\rank P_1$, then there is a unitary $U\colon \bC^k\to \bC^{k_1}\oplus \bC^{k_2}$ for which 
$UA_p=
\begin{sbmatrix}
A_p'\\0
\end{sbmatrix}$ for some $k_1\times n$ matrices $A_p'$, and $UB_p=
\begin{sbmatrix}
  0\\B_p'
\end{sbmatrix}$ for some $k_2\times n$ matrices $B_p'$. 
Now $C_p^*C_q=(UC_p)^*(UC_q)=
\begin{sbmatrix}
  {A_p'}^*A_q'&0\\0&{B_p'}^*B_q'
\end{sbmatrix}$, and it follows that $\gamma(\S_i)\leq k_i$ for
$i=1,2$. Hence, $\gamma(\S_1)+\gamma(\S_2)\leq k_1+ k_2=k=\gamma(\S)$.
Combined with (iv), this shows that $\gamma(\S_1\oplus \S_2) = \gamma(\S_1)+\gamma(\S_2)$.
\end{proof}

\begin{remark}\label{r_inco}
{\bf (i) } Let $\pi\in\{\alpha,\beta,\gamma\}$ and $d\in \bb{N}$. Then 
$\pi(M_d(\cl S)) = \pi(\cl S)$.  Indeed, by Proposition \ref{prop:elem} (ii), 
we have $\pi(\cl S) \leq \pi(M_d(\cl S))$, and 
the reverse inequality for $\pi\in\{\beta,\gamma\}$ follows from Proposition~\ref{prop:elem}~(iii) and Remark~\ref{remark:pi=1}~(iii). To see the corresponding result for $\pi=\alpha$,
suppose that $\{\xi_p\}_{p=1}^m$ is an independent set for $\cl S\otimes M_d$. 
Then for $X,Y\in M_d$, $A\in \cl S$ and $p\neq q$, we have
\begin{equation}\label{eq_xya}
\langle (A\otimes I)(I\otimes X)\xi_p, (I\otimes Y)\xi_q\rangle = 0
\end{equation}
Let $Q_p$ be the projection onto $\spn\left\{(I\otimes X)\xi_p : X\in M_d\right\}$; then $Q_p = E_p\otimes I_d$
for some non-zero projection $E_p$ on $\bb{C}^n$, and 
(\ref{eq_xya}) implies that $E_q\cl S E_p = \{0\}$ provided $p\neq q$.
If $v_p$ is a unit vector with $E_p v_p = v_p$, $p\in [m]$, we therefore have that $\{v_p\}_{p=1}^m$ is an 
independent set for $\cl S$. It follows that $\alpha(\cl S) \geq \alpha(M_d(\cl S))$ and hence we have equality. 

\smallskip

\noindent{\bf (ii) } 
The parameter  $\gamma$ is neither order-preserving nor order-reversing for
  inclusion. For example, $\bC I_2\subseteq \S\subseteq M_2$ where
  $\S$ is the operator system of~Proposition~\ref{prop:n=2,gamma=3}, and
  these operator systems have $\gamma$-values $2,3,1$, respectively.
\end{remark}

We will now show that like its graph-theoretic counterpart, namely, the minimum semidefinite rank, 
$\gamma(\S)$ is the solution to a rank minimisation problem.

\begin{proposition}\label{p_eta}
  For any operator system~$\S\subseteq M_n$, we have
  \begin{multline*}
    \gamma(\S) = \min_{m\in \bb{N}} \biggr\{ \rank B \colon B=[B_{i,j}] \in M_m(\S)^+ \text{ with } \\
    \ifieee\else\hphantom{\iff\iff}\fi\spn \{B_{i,j}\colon i,j\in [m]\}=\S\text{ and }\sum_{i=1}^m B_{i,i}=I_n \biggr\}
  \end{multline*}
  and
\ifieee\begin{multline*}
    \beta(\cl S) = \min_{m\in \bb{N}} \biggr\{ \rank B : B=[B_{i,j}] \in M_m(\cl S)^+\ifieee\\\fi \text{ and } \sum_{i=1}^m B_{i,i} = I_n \biggr\}.
\end{multline*}
\else
\[
    \beta(\cl S) = \min_{m\in \bb{N}} \biggr\{ \rank B : B=[B_{i,j}] \in M_m(\cl S)^+ \text{ and } \sum_{i=1}^m B_{i,i} = I_n \biggr\}.
\]
\fi
Moreover, the minima on the right hand sides are achieved for~$m$ not exceeding $2n^3$. 
\end{proposition}

\begin{proof}
  Suppose that $m\in \bN$ and that 
  $B = [B_{i,j}]\in M_m(\S)^+$ satisfies the relations
  $\spn \{B_{i,j}\}_{i,j=1}^m = \cl S$ and $\sum_{i=1}^m B_{i,i} = I_n$. 
  Then $B=A^*A$ for some
  $A=[A_1\ \dots\ A_m]\in M_{1,m}(M_{k,n})$ where $k=\rank B$. 
  Since
  $B_{i,j}=A_i^*A_j$, we see that $\{A_1,\dots,A_m\}$ are Kraus
  operators for a quantum channel~$\Phi$ with $\S_\Phi=\S$; thus,
  $\gamma(\S)\leq k$.

  Conversely, let $k=\gamma(\S)$, $m\in \bN$ and
  $A_1,\dots,A_m\in M_{k,n}$ be Kraus operators for a quantum
  channel~$\Phi$ with $\S_\Phi=\S$.
  Set $B:=[A_i^*A_j]\in M_m(\S)^+$; we have 
  $\spn \{ A_i^*A_j: i,j \in [m] \} = \cl S$, $\sum_{i=1}^m A_i^*A_i = I_n$ and $\rank B=\rank [A_1\ \dots\ A_m]\leq k$.
  Hence the minimum rank in the first expression is no greater than $\gamma(\S)$.
  
To see that some $m\le 2n^3$ attains this minimum, set $k = \gamma(\cl S)$. 
Then there exists a quantum channel $\Phi : M_n\to M_k$ with $\cl S_{\Phi} = \cl S$ and,
by Corollary~\ref{c_bong}, $k\le 2n^2$.
By~\cite[Remark~6]{choi}, the channel $\Phi$ can be realised using at most $nk\leq 2n^3$ Kraus operators. Since $m$ is precisely the number of Kraus operators in the preceding argument, we see that 
the minimum in the expression for $\gamma(\S)$ is attained for some $m\leq 2n^3$. 

The expression for $\beta(\cl S)$ follows from the fact that 
$\beta(\cl S) = \min \{ \gamma(\cl T): \cl T \subseteq \cl S \}.$
Since the bound on $m$ for $\gamma$, namely $2n^3$, is independent of the operator system $\cl S\subseteq M_n$, this fact shows that here we may also take $m\le 2n^3$.
\end{proof}

Our next task is to show that the operator system parameters just defined generalise the graph parameters of Section~\ref{sec:graphs}. 
Recall that if $G$ is a graph with vertex set~$[n]$, we let 
$$\S_G = \spn\{E_{i,j}\colon i\simeq j\}$$
be the associated operator subsystem of $M_n$. 

\begin{lemma}\label{lem:Delta_x}
  Let $n,k\in\bN$, and let $\Delta$ be the group of diagonal $n \times n$ matrices whose diagonal
entries are each either $1$ or $-1$. Let $x=(x_1,\dots,x_n)$ be an $n$-tuple of non-zero vectors in~$\bC^k$, and let 
$A = [\hat x_1~\cdots~\hat x_n]$ be the
$k\times n$ matrix whose $i$-th column is the unit vector $\hat x_i=\|x_i\|^{-1}x_i$. Then the map
\[\Delta_x\colon M_n\to M_k,\quad \Delta_x(X)=2^{-n}\sum_{D\in \Delta} ADXDA^*,\]
is a quantum channel. Moreover, if $x_i \in \bR_+^k\setminus\{0\}$, $i = 1,\dots,n$, then $\Delta_x$ is non-cancelling.
\end{lemma}

\begin{proof}
For $D\in \Delta$, let $d_i\in\{1,-1\}$ be the $i$-th diagonal entry of $D$. We have
\begin{align*}
  2^{-n}\sum_{D\in \Delta}(DA^*AD)_{ij}
  =2^{-n}\langle \hat x_j,\hat x_i\rangle\sum_{D\in \Delta}d_{i} d_{j}=\delta_{ij},
\end{align*}
since if $i\ne j$ then the sum reduces to~$0$ by symmetry, whereas
if~$i=j$ then every term in the sum is~$1$. Since each~$\hat x_i$ is a unit vector, we obtain $2^{-n}\sum_{D\in \Delta} DA^*AD=I_n$, so $\Delta_x$ is a quantum channel. The assertion about non-cancelling channels follows trivially.
\end{proof}

\begin{proposition}\label{pr_commsq}
  Let $n,k\in \bN$, $x_i$ be a non-zero vector in~$\bC^k$, $i= 1,\dots,n$, and 
  $x = (x_1,\dots,x_n)$. Then $\S_{G(x)} = \S_{\Delta_x}$. 
\end{proposition}

\begin{proof}
  Let
  $\S=\S_{G(x)}$ and $\T=\S_{\Delta_x}$. 
  Set $\hat x_i=\|x_i\|^{-1}x_i$ and $A=[\hat x_1~\cdots~\hat x_n]$, and note 
that $\T$ is spanned
  by the operators $DA^*AD'$ for $D=\diag(d_1,\dots,d_n)$ and
  $D'=\diag(d'_1,\dots,d'_n)$ in $\Delta$.  
  For $i,j\in [n]$, we have
  \begin{align*} \sum_{\substack{D,D'\in \Delta,\\d_i=d'_j=1}}DA^*AD' &=
  \Big(\sum_{\substack{D\in \Delta,\\d_{i}=1}}D\Big)A^*A
  \Big(\sum_{\substack{D'\in \Delta,\\d'_{j}=1}}D'\Big)\\
  &= 4^{n-1}
    E_{i,i}A^*A E_{j,j}\\
  &= 4^{n-1}\langle \hat x_j,\hat x_i\rangle E_{i,j}.
\end{align*}
If $E_{i,j}\in \S$, then $i\simeq_{G(x)}j$, so
$\langle \hat x_j,\hat x_i\rangle\ne 0$, hence $E_{i,j}\in \T$. Thus $\S\subseteq \T$. 
On the other hand, $\S^\perp$ is spanned by the matrix units $E_{i,j}$ with $i\not\simeq_{G(x)} j$. 
For such $i,j$ and any~$D,D'\in \Delta$, we have $(DA^*AD')_{ji}=d_j\langle \hat x_i,\hat x_j\rangle d'_i=0$, 
so $E_{i,j}\in \T^\perp$. Hence $ \S^\perp\subseteq \T^\perp$ as required.
\end{proof}

\begin{theorem}\label{theorem:gen}
  For any graph $G$ with vertex set~$[n]$ and $\pi\in\{\alpha,\beta,\gamma,\cc\}$, we have $\pi(\S_G)=\pi(G)$.
\end{theorem}
\begin{proof}
  The case $\pi=\alpha$ is known~\cite{dsw}.

We next consider the case $\pi = \gamma$. If
  $k = \gamma(G)$, then there exists $x = (x_1,\dots,x_n)$,
  where $x_i\in\bC^k\setminus\{0\}$, $i = 1,\dots,n$, so that $G(x) = G$, and hence
  $\S_{G(x)} = \S_G$. By Proposition~\ref{pr_commsq}, $\S_{G(x)}$
  is the operator system of a quantum channel $M_n \to M_k$, hence
  $\gamma(\S_G)\leq k = \gamma(G)$.

  Now let $k = \gamma(\S_G)$, so that there are Kraus operators
  $A_1,\dots,A_m\in M_{k,n}$ for a quantum channel
  $\Phi\colon M_n\to M_k$ with
  $\S_\Phi = \spn\{A_q^*A_p\colon p,q\in [m]\} = \S_G$. 
  Since the column operator with entries $A_1,\dots,A_m$ is an isometry, for each $i\in [n]$ we have
  $\sum_{p=1}^m\|A_p e_i\|^2=\|e_i\|=1$. In particular, $A_pe_i\ne 0$ for at least one $p\in [m]$.
  Thus, the projection $P_i\in M_k$ onto the span of
  $\{A_pe_i\colon p\in [m]\}$ is non-zero. 
  
  Consider the tuple $P = (P_1,\dots,P_n)\in \P(k)$.
   Since $\S_\Phi=\S_G$, for $i,j\in [n]$ we have
\ifieee
  \begin{align*}
    i\not\simeq_{G(P)}j & \iff P_iP_j=0\\&\iff \langle A_pe_i,A_q e_j\rangle=\tr(E_{ij}^*A_q^*A_p)=0,\\&\hphantom{\iff\iff\iff\iff\iff}\text{for all $p,q\in[m]$}\\
    &\iff E_{ij}\in \S_\Phi^\perp=\S_G^\perp \iff i\not\simeq_G j.
  \end{align*}
\else
  \begin{align*}
    i\not\simeq_{G(P)}j & \iff P_iP_j=0\\&\iff \langle A_pe_i,A_q e_j\rangle=\tr(E_{ij}^*A_q^*A_p)=0\;\text{for all $p,q\in[m]$}\\
    &\iff E_{ij}\in \S_\Phi^\perp=\S_G^\perp \iff i\not\simeq_G j.
  \end{align*}
\fi
  Hence $G(P)=G$. Using Theorem~\ref{thm:qinter=mpsd}, we obtain
  $\gamma(G)=\qinter(G)\leq k=\gamma(\S_G)$.

The case $\pi = \beta$ is similar to the case $\pi = \gamma$; alternatively, see \cite[Theorem 12]{stahlke}.

Let $\pi =  \cc$.  
Set $k=\cc(G)$;
  then there exists a tuple $x=(x_1,\dots,x_n)\in(\bR_+^k\setminus\{0\})^n$ with $G(x)=G$. 
  By Proposition \ref{pr_commsq}, 
  $\S_{\Delta_x} = \S_{G(x)} = \S_G$. 
  By Lemma~\ref{lem:Delta_x}, $\Delta_x$ is a non-cancelling quantum channel; therefore, $\cc(\S_G)\leq k$. 
  In particular, $\cc(\S_G)<\infty$. Now if $k=\cc(\S_G)$, let $\Phi\colon M_n\to M_k$ 
  be a non-cancelling quantum channel such that $\S_\Phi=\S_G$. Fix Kraus operators $A_1D_1,\dots,A_mD_m$ for $\Phi$ where each $A_p\in M_{k,n}$ is entrywise non-negative and each $D_p$ is an invertible diagonal $n\times n$ matrix. For $i\in [n]$, define
  \[R_i = \left\{r\in [k]\colon \langle A_pe_i,e_r\rangle\ne 0\text{ for some
    $p\in[m]$}\right\}.\]
Since the column operator with entries $A_1D_1,\dots,A_mD_m$ is an isometry, we have 
$\sum_{p=1}^m \|A_pD_pe_i\|^2=1$, and so $R_i$ is non-empty for all
$i \in [n]$. 
Let $x=(x_1,\dots,x_n)$ where $x_i=\sum_{r\in R_i}e_r$.
  For $i,j\in[n]$, we have
  \begin{align*}
    i\simeq_G j&\iff E_{j,j}\S_G E_{i,i}\ne \{0\}\\
    &\iff \exists\, p,q\in [m] \mbox{ such that } \langle A_pD_pe_{i},A_qD_q e_{j}\rangle\ne 0\\
    &\iff \exists\, p,q\in [m]  \mbox{ such that }  \langle A_pe_i,A_qe_{j}\rangle\ne 0.
  \end{align*}
  Now \[\langle A_pe_i,A_qe_{j}\rangle=\sum_{r\in [k]} \langle A_pe_i,e_r\rangle \langle e_r,A_qe_{j}\rangle
  \]
  and every term in the latter sum is non-negative. It follows that
\ifieee
  \begin{align*}
    i\simeq_Gj&\iff \exists\,p,q\in[m],\; r\in[k] \mbox{ with } \langle A_pe_i,e_r\rangle\ne 0\\&\hphantom{\iff\exists\,p,q\in[m],\; r\in[k]\;\;}\text{ and } \langle A_q e_{j},e_r\rangle\ne 0\\
    &\iff R_{i}\cap R_{j}\ne \emptyset\iff \langle x_{j},x_i\rangle\ne 0\\&\iff i\simeq_{G(x)}j.
  \end{align*}
\else
  \begin{align*}
    i\simeq_Gj&\iff \exists\,p,q\in[m],\; r\in[k] \mbox{ with } \langle A_pe_i,e_r\rangle\ne 0\text{ and } \langle A_q e_{j},e_r\rangle\ne 0\\
    &\iff R_{i}\cap R_{j}\ne \emptyset\iff \langle x_{j},x_i\rangle\ne 0\iff i\simeq_{G(x)}j.
  \end{align*}
\fi
  Hence $G = G(x)$, so $\cc(G)\leq k = \cc(\S_G)$.
\end{proof}

\begin{remark} \label{rk:classical-quantum}
Let $(p(y|x))$ be a $k \times n$ non-negative column-stochastic matrix defining a 
classical channel $\cl N: [n] \to [k]$ with confusability graph $G= G_{\cl N}$. 
The canonical quantum channel $\cl N_q : M_n \to M_k$ associated with $\cl N$ is defined by setting
$\cl N_q(E_{x,x}) = \sum_{y\in [k]} p(y|x) E_{y,y}$ and $\cl N_q(E_{x,x'}) = 0$ if $x \ne x'$.
We have that $\cl S_{\cl N_q} = \cl S_G$ \cite{dsw} (see also \cite{p_notes}).  
So we see that  $\gamma(G) = \gamma(\cl S_G)$ is the quantum complexity of the classical 
channel $\cl N$ when viewed as a quantum channel.
\end{remark}

Let $G\boxtimes H$ denote the {\it strong product} of the graphs $G$ and $H$~\cite{sa}, in which
$(x,y)\simeq_{G \boxtimes H} (x',y')$ if and only if $x\simeq_{G} x'$ and $y\simeq_{H} y'$.
Note that $\cl S_{G \boxtimes H} = \cl S_G \otimes \cl S_H$. 
If~$G,H$ are graphs and $n$ is the number of vertices of~$G$, then
\begin{itemize}
\item[(i)] $\alpha(G)=1$ if and only if $G=K_n$, i.e.,~if and only if $\S_G = M_n$;
\item[(ii)] $\gamma(G)\leq n$; and 
\item[(iii)] $\gamma(G\boxtimes H)\leq\gamma(G)\gamma(H)$, but it is unknown whether strict inequality can occur. 
\end{itemize}
The following proposition shows that the parameters for general operator systems $\S\subseteq M_n$
behave quite differently, with respect to the latter properties, 
than their graph theoretic counterparts.

\begin{proposition}\label{prop:n=2,gamma=3}
  Let $\S=\left\{
        \begin{sbmatrix}
          \lambda&a\\{}b&\lambda
        \end{sbmatrix}\colon a,b,\lambda\in\bC\right\}$. Then
     $\alpha(\S)=1$, $\beta(\S)=2$, $\gamma(\S)=\cc(\S)=3$ and $\gamma(\S\otimes\S)<\gamma(\S)^2$.
     \end{proposition}

   \begin{proof}
  The Kraus operators $A_1=
  \frac1{\sqrt2}\begin{sbmatrix}
    1&0\\0&0\\0&1
  \end{sbmatrix}$ and $A_2=\frac1{\sqrt2}
  \begin{sbmatrix}
    0&0\\0&1\\1&0
  \end{sbmatrix}$ yield a non-cancelling quantum channel with operator system~$\S$,
  so $\gamma(\S)\leq\cc(\S)\le3$. Since $\S^\perp$ is spanned
  by $
  \begin{sbmatrix}
    1&0\\0&-1
  \end{sbmatrix}$, it contains no rank one operators, and hence
  $\alpha(\S)=1$. By Remark \ref{remark:pi=1} (v), $\beta(\S)\leq 2$ while, 
  by Remark~\ref{remark:pi=1} (iii),
  $\beta(\S)\ne 1$; thus, $\beta(\S)=2$.

  If $\gamma(\S)\le 2$, then there are $A_j=[v_j\ w_j]\in M_2$ for
  some $v_j,w_j\in \bC^2$ which are the Kraus operators of a quantum
  channel with operator system~$\S$. We have
  $\langle v_j,v_i\rangle = (A_i^*A_j)_{11}=(A_i^*A_j)_{22}=\langle
  w_j,w_i\rangle$ for each $i,j$. It follows that there is a
  $2\times 2$ unitary $U$ with $Uv_j=w_j$, $j = 1,2$. Recall that
$\sum_{i=1}^2 A_i^*A_i=I_2$. This is equivalent to 
$$\sum_{i=1}^2  \|v_i\|^2 = \sum_{i=1}^2  \|w_i\|^2 = 1 \ifieee$$and$$\else\ \mbox{ and }  \ \fi
\sum_{i=1}^2  \langle w_i,v_i\rangle = \sum_{i=1}^2  \langle Uv_i,v_i\rangle = 0.$$ 
In particular, $0$ lies in the
  numerical range of~$U$. However, the numerical range of a normal matrix
  is the convex hull of its spectrum, and hence
  $\sigma(U)=\{ \alpha,-\alpha\}$ for some $\alpha\in \mathbb{T}$. 
  Thus
  $\overline {\alpha} U$ is hermitian, and so $U = \alpha^2U^*$. Now
\ifieee
  \begin{align*}
    A_i^*A_j&=\left(\begin{smallmatrix} \langle v_j,v_i\rangle&\langle
        Uv_j,v_i\rangle\\\langle U^* v_j,v_i\rangle & \langle
        Uv_j,Uv_i\rangle \end{smallmatrix}\right) \\&= \langle
    v_j,v_i\rangle I + \langle U^*v_j,v_i\rangle
    \left(\begin{smallmatrix} 0&\alpha^2 \\1&0
      \end{smallmatrix}\right),
  \end{align*}
\else
  \begin{align*}
    A_i^*A_j&=\left(\begin{smallmatrix} \langle v_j,v_i\rangle&\langle
        Uv_j,v_i\rangle\\\langle U^* v_j,v_i\rangle & \langle
        Uv_j,Uv_i\rangle \end{smallmatrix}\right) = \langle
    v_j,v_i\rangle I + \langle U^*v_j,v_i\rangle
    \left(\begin{smallmatrix} 0&\alpha^2 \\1&0
      \end{smallmatrix}\right),
  \end{align*}
\fi
  hence
  $3=\dim\S=\dim\spn\{A_i^*A_j\}\le 2$, a contradiction.

  To see that~$\gamma(\S\otimes \S)<\gamma(\S)^2=9$, consider the isometries
  $V_i\in M_{8,4}$ given by
\ifieee
  \begin{align*}
    V_1&=[e_1\ e_2\ e_3\ e_4],\quad V_2=[e_2\ e_5\ e_4\ e_6],\\
    V_3&=[e_3\ e_4\ e_7\ e_8],\quad V_4=[e_6\ e_7\ e_5\ e_1].
  \end{align*}
\else
  \begin{align*}
    V_1&=[e_1\ e_2\ e_3\ e_4],\quad V_2=[e_2\ e_5\ e_4\ e_6],\quad
  V_3=[e_3\ e_4\ e_7\ e_8],\quad V_4=[e_6\ e_7\ e_5\ e_1].
  \end{align*}
\fi
  Let $W=\tfrac12[V_1\ V_2\ V_3\ V_4]\in M_{8,16}$; then
\[\newcommand{\eye}{\raisebox{-.6ex}{\text{\Large$I_4$}}}\newcommand{\oo}{0}W^*W=
  \frac14\left[\begin{smallmatrix}
      \eye
      &
      \begin{smallmatrix}\oo&\oo&\oo&\oo\\1&\oo&\oo&\oo\\\oo&\oo&\oo&\oo\\\oo&\oo&1&\oo\end{smallmatrix}
      &
      \begin{smallmatrix}\oo&\oo&\oo&\oo\\\oo&\oo&\oo&\oo\\1&\oo&\oo&\oo\\\oo&1&\oo&\oo\end{smallmatrix}
      &
      \begin{smallmatrix}\oo&\oo&\oo&1\\\oo&\oo&\oo&\oo\\\oo&\oo&\oo&\oo\\\oo&\oo&\oo&\oo\end{smallmatrix}
      \\[\medskipamount]
      \begin{smallmatrix}\oo&1&\oo&\oo\\\oo&\oo&\oo&\oo\\\oo&\oo&\oo&1\\\oo&\oo&\oo&\oo\end{smallmatrix}
      &
      \eye
      &
      \begin{smallmatrix}\oo&\oo&\oo&\oo\\\oo&\oo&\oo&\oo\\\oo&1&\oo&\oo\\\oo&\oo&\oo&\oo\end{smallmatrix}
      &
      \begin{smallmatrix}\oo&\oo&\oo&\oo\\\oo&\oo&1&\oo\\\oo&\oo&\oo&\oo\\1&\oo&\oo&\oo\end{smallmatrix}
      \\[\medskipamount]
      \begin{smallmatrix}\oo&\oo&1&\oo\\\oo&\oo&\oo&1\\\oo&\oo&\oo&\oo\\\oo&\oo&\oo&\oo\end{smallmatrix}
      &
      \begin{smallmatrix}\oo&\oo&\oo&\oo\\\oo&\oo&1&\oo\\\oo&\oo&\oo&\oo\\\oo&\oo&\oo&\oo\end{smallmatrix}
      &
      \eye
      &
      \begin{smallmatrix}\oo&\oo&\oo&\oo\\\oo&\oo&\oo&\oo\\\oo&1&\oo&\oo\\\oo&\oo&\oo&\oo\end{smallmatrix}
      \\[\medskipamount]
      \begin{smallmatrix}\oo&\oo&\oo&\oo\\\oo&\oo&\oo&\oo\\\oo&\oo&\oo&\oo\\1&\oo&\oo&\oo\end{smallmatrix}
      &
      \begin{smallmatrix}\oo&\oo&\oo&1\\\oo&\oo&\oo&\oo\\\oo&1&\oo&\oo\\\oo&\oo&\oo&\oo\end{smallmatrix}
      &
      \begin{smallmatrix}\oo&\oo&\oo&\oo\\\oo&\oo&1&\oo\\\oo&\oo&\oo&\oo\\\oo&\oo&\oo&\oo\end{smallmatrix}
      &
      \eye
    \end{smallmatrix}\right]
  \]

  Writing
  $B_{i,j}$ for the $(i,j)$-th $4\times4$ block of~$W^*W$, we observe that
  \[\spn\left\{B_{i,j}\colon i,j\in [4]\right\} = \S\otimes \S\qand \sum_{i=1}^4B_{i,i}=I_4.\] 
  By Proposition~\ref{p_eta}, 
  \ifieee
  \[\gamma(\S\otimes \S)\leq \rank W^*W=\rank W= 8.\qedhere\]
  \else$\gamma(\S\otimes \S)\leq \rank W^*W=\rank W= 8$.\fi
\end{proof}

\section{Applications to capacity}
\label{sec:applications-to-capacity}

Let $\cl N : X\to Y$ be a classical information channel with confusability graph $G$.
Its parallel use $r$ times can be expressed as a channel 
$\cl N^{\times r}: X^r \to Y^r$, for which 
$$p((y_s)_{s=1}^r |(x_s)_{s=1}^r) = \prod_{s=1}^r p(y_s|x_s),\ifieee$$for $x_s\in X$, $y_s\in Y$ and $s = 1,\dots,r$. \else\ \ \ x_s\in X,\; y_s\in Y,\; s = 1,\dots,r.$$\fi
Note that
$$G_{\cl N^{\times r}} = \underbrace{G_{\cl N} \boxtimes \cdots \boxtimes G_{\cl N}}_{\text{$r$ times}}.$$
The {\it Shannon capacity} of the channel $\cl N : X\to Y$ 
(or equivalently of the graph $G$) is the quantity
\[ \Theta(\cl N) = \Theta(G) = \lim_{r\to \infty} \sqrt[r]{\alpha(\cl N^{ \times r})} = \lim_{r\to \infty} \sqrt[r]{\alpha(G_{\cl N}^{\boxtimes r})}.\]
(Some authors prefer to use the logarithm of the quantities defined above.)

Similarly, if $\Phi : M_n\to M_k$ is a quantum channel, letting $\Phi^{\otimes r} : M_n^{\otimes r} \to M_k^{\otimes r}$ 
be its $r$-th power, we find 
$$\cl S_{\Phi^{\otimes r}} = \underbrace{\cl S_{\Phi} \otimes \cdots \otimes \cl S_{\Phi}}_{\text{$r$ times}}.$$
The analogue of the Shannon capacity of a quantum channel introduced in 
\cite{dsw} is the parameter
\[ \Theta(\Phi) = \lim_{r\to\infty} \sqrt[r]{\alpha(\Phi^{\otimes r})}.\]

Lov\'asz \cite{lo} introduced his famous $\vartheta$-parameter of a graph and proved that
$\alpha(\cl N) \le \vartheta(G)$ and that $\vartheta$ is multiplicative for strong graph product; hence, 
\[ \Theta(\cl N) = \Theta(G_{\cl N}) \le \vartheta(G),\]
for any classical channel,
thus giving a bound on the Shannon capacity of classical channels.  He also proved \cite[Theorem~11]{lo} that
\[ \vartheta(G) \le \beta(G),\]
so that his $\vartheta$-bound is a better bound on the capacity of classical channels than any of the bounds that we derived from complexity considerations.
However, as we will shortly show, for quantum channels, $\beta$ yields a bound on capacity that can outperform $\vartheta$.
We note that a different bound on $\Theta(G)$, based on ranks of Hermitian matrices in the operator system $\cl S_{G^c}$, 
was introduced by Haemers in \cite{haemers}. 
It is an interesting open question to formulate general non-commutative analogues of Haemers' parameter. 

Lov\'asz gave many characterisations of his parameter, but the most useful for our purposes is the expression
\[ \vartheta(G) = \max \left\{ \|I +K \|:  I+K \in M_n^+,\;  K \in \cl S_G^{\perp}\right\}.\]
The latter formula motivated \cite{dsw} to define, for any operator subsystem $\cl S$ of $M_n$,
\[ \vartheta(\cl S) = \max \left\{ \|I +K \|: I+K \in M_n^+,  \; K\in \cl S^\perp\right\};\] 
note that $\vartheta(G) = \vartheta(\cl S_G)$.
It was shown in \cite{dsw} that, for any quantum channel $\Phi$, one has 
\[ \alpha(\Phi) = \alpha(\cl S_{\Phi}) \le \vartheta(\cl S_{\Phi}).\]
However, $\vartheta$ is only supermultiplicative for tensor products of general operator systems.  This motivated \cite{dsw} to introduce a ``complete'' version, denoted $\thetatilde$, which is multiplicative for tensor products of operator systems and satisfies $\vartheta(\cl S) \le \thetatilde(\cl S)$.
This allowed them to bound the quantum capacity of a quantum channel, since
\ifieee
\begin{align*} \Theta(\Phi) = \lim_{r\to\infty} \sqrt[r]{\alpha\left(\cl S_{\Phi}^{\otimes r}\right)} 
&\le \lim_{r\to\infty} \sqrt[r]{\vartheta\left(\cl S_{\Phi}^{\otimes r}\right)} 
\\&\le \lim_{r\to\infty} \sqrt[r]{\thetatilde\left(\cl S_{\Phi}^{\otimes r}\right)} = \thetatilde\left(\cl S_{\Phi}\right).\end{align*}
\else
\begin{align*} \Theta(\Phi) = \lim_{r\to\infty} \sqrt[r]{\alpha\left(\cl S_{\Phi}^{\otimes r}\right)} 
&\le \lim_{r\to\infty} \sqrt[r]{\vartheta\left(\cl S_{\Phi}^{\otimes r}\right)} 
\le \lim_{r\to\infty} \sqrt[r]{\thetatilde\left(\cl S_{\Phi}^{\otimes r}\right)} = \thetatilde\left(\cl S_{\Phi}\right).\end{align*}
\fi
These bounds are often difficult to compute. The quantity $\lim_{r\to \infty}\sqrt[r]{\vartheta(\S_\Phi^{\otimes r})}$ requires evaluation of a limit, each term of which may be intractable, and the possibly larger bound $\thetatilde(\S_\Phi)=\sup_{n\in\bN}\vartheta(\S_\Phi\otimes M_n)$ 
requires the evaluation of a supremum, although this parameter has the advantage of possessing a reformulation as a semidefinite program~\cite{dsw}.

\begin{theorem}\label{th_boundbeta} 
For any quantum channel $\Phi$, we have 
$$\alpha(\Phi)\leq \Theta\left(\Phi\right)  \le \beta\left(\cl S_{\Phi}\right).$$
\end{theorem}

\begin{proof}
Let $\Phi$ be a quantum channel and $\cl S = \cl S_{\Phi}$. 
The inequality $\alpha\leq \Theta$ is well known, and follows immediately from the supermultiplicative property of $\alpha$. 
Since $\beta$ is submultiplicative for tensor products (Proposition \ref{prop:elem} (iii)) 
and $\alpha$ is dominated by $\beta$ (Theorem \ref{thm:inncc}), 
we have 
\[ \Theta\left(\Phi\right) = \Theta(\cl S) = \lim_{r\to\infty} \sqrt[r]{\alpha\left(\cl S^{\otimes r}\right)} \le \beta\left(\cl S\right).\qedhere\]
\end{proof}

In the remainder of the section, we will exhibit operator systems for which $\beta(\cl S) \ll \vartheta(\cl S)$.
For $k\in \bb{N}$, let 
$$\cl S_k = \{(a_{i,j})_{i,j=1}^k \in M_k : a_{1,1} = a_{2,2} = \cdots = a_{k,k}\}.$$
It is easy to show directly that \[\vartheta(\cl S_k) =k.\]
For any $m\in \bb{N}$, applying the canonical shuffle which identifies $M_k\otimes M_m$ 
with $M_m\otimes M_k$, we have
\ifieee
\begin{multline*}
  \cl S_k \otimes M_m = \left\{(A_{i,j})_{i,j=1}^k \in M_k(M_m) :\right.\\\left.\vphantom{(A_{i,j})_{i,j=1}^k} A_{1,1} = A_{2,2} = \cdots = A_{m,m}\right\}.
\end{multline*}
\else
\[
  \cl S_k \otimes M_m = \left\{(A_{i,j})_{i,j=1}^k \in M_k(M_m) : A_{1,1} = A_{2,2} = \cdots = A_{m,m}\right\}.
\]
\fi
Thus, for any operator system $\cl S\subseteq M_m$, we have 
\ifieee
\begin{multline*}\cl S_k \otimes \cl S =  \left\{(A_{i,j})_{i,j=1}^k \in M_k(\cl S) : \right.\\\left.\vphantom{(A_{i,j})_{i,j=1}^k} A_{1,1} = A_{2,2} = \cdots = A_{m,m}\right\}.\end{multline*}\else
\[\cl S_k \otimes \cl S =  \left\{(A_{i,j})_{i,j=1}^k \in M_k(\cl S) :  A_{1,1} = A_{2,2} = \cdots = A_{m,m}\right\}.\]\fi

\begin{theorem}\label{th_betabetter}
We have 
\ifieee\[\else$\fi\beta(\cl S_k\otimes\cl S_{k^2}) \leq k^2 < k^3 \leq \vartheta(\cl S_k\otimes \cl S_{k^2})
\ifieee.\]\else$.\fi
\end{theorem}
\begin{proof}
Let $\omega$ be a primitive $k$-th root of unity. 
Let $S \in M_{k^2}$ be given by $Se_i = e_{i+1}$, $i = 1,\dots,k^2$, where addition is modulo $k^2$, 
while $D \in M_{k^2}$ be the diagonal matrix with diagonal $(1,\omega,\omega^2,\dots,\omega^{k^2-1})$. 
Note that, if $D_0 \in M_k$ is the diagonal matrix with diagonal $(1,\omega,\omega^2,\dots,\omega^{k-1})$,
then $D = \underbrace{D_0 \oplus\cdots \oplus D_0}_{\text{$k$ times}}$. 
We have that $D^j S = \omega^j S D^j$ for any $j\in \bb{Z}$, and hence 
\begin{equation}\label{eq_sd}
D^j S^i = \omega^{ij} S^i D^j, \ \ \ i,j\in \bb{Z}.
\end{equation}
Any element $A$ of $M_k(\cl S_k\otimes \cl S_{k^2})$ has the form 
$A = (C_{r,s})_{r,s=0}^{k-1}$, where $C_{r,s}\in \cl S_k\otimes \cl S_{k^2}$ for all $r,s = 0,\dots,k-1$. 
In view of the remarks before the statement of the theorem,
we may write $C_{r,s} = (A_{kr + i, ks + j})_{i,j=0}^{k-1}$, where
$A_{kr + i, ks + j}\in \cl S_{k^2}$ for all $r,s,i,j = 0,\dots,k-1$, and 
$$A_{kr+1,ks+1} = A_{kr+2,ks+2} = \cdots = A_{(r+1)k,(s+1)k}, \ifieee$$ for all $r,s=1,\dots,k$.\else\ \ \mbox{ for all } 
r,s = 0,\dots,k-1.$$\fi

Let 
$$u_{kr + i} = S^{kr+i} D^r, \ \ \ r,i = 0,\dots,k-1,$$
and
$$B = (u_0,u_1,u_2,\dots,u_{k^2 -1}) \in M_{k^2,k^4}.$$
Set $B_{r,s} = (u_{kr+i}^* u_{ks+j})_{i,j=0}^{k-1}$; 
then the matrix
$$B^*B = (B_{r,s})_{r,s=0}^{k-1} = (u_{kr+i}^* u_{ks+j})_{r,s,i,j}$$ is positive and has rank at most $k^2$. 
We will show the following:
\smallskip
\begin{itemize}
\item[(i)]  $u_{kr+i}^* u_{ks+j}\in \cl S_{k^2}$, for all $r,s,i,j = 0,1,\dots,k-1$;
\smallskip
\item[(ii)] $u_{kr+1}^* u_{ks+1} = u_{kr+2}^* u_{ks+2} = \cdots = u_{kr+k-1}^* u_{ks+k-1}$, for all $r,s,i = 0,\dots,k-1$, and 
\smallskip
\item[(iii)] $\sum_{r=1}^k B_{r,r} = k I$,
\end{itemize}
\smallskip
which will imply that $\beta(\cl S_k\otimes \cl S_{k^2}) \leq k^2$. 

To show (i), note that 
$$u_{kr+i}^* u_{ks+j} = D^{-r} S^{-kr-i} S^{ks+j} D^s = D^{-r} S^{ks -kr + j - i} D^s.$$
If $ks -kr + j - i\neq 0$, then $u_{kr+i}^* u_{ks+j}$ has zero diagonal and thus belongs to $\cl S_{k^2}$. 
Suppose that $ks -kr + j - i = 0$. Then $k | (i-j)$ and hence $i = j$. If, in addition, $r \neq s$ then 
$u_{kr+i}^* u_{ks+j}$ has zero diagonal and therefore belongs to $\cl S_{k^2}$; 
if, on the other hand, $r = s$, then $u_{kr+i}^* u_{ks+j} = I$ and hence again belongs to $\cl S_{k^2}$. 

To show (ii), note that \ifieee for $i=0,\dots,k-1$, we have\fi
$$u_{kr+i}^* u_{ks+i} = D^{-r} S^{-kr-i} S^{ks+i} D^s = D^{-r} S^{k(s-r)} D^s\ifieee\else, \ \  i = 1,\dots,k-1\fi.$$

In order to show (iii), suppose that $i,j\in \{0,\dots,k-1\}$ with $i \neq j$, and, using (\ref{eq_sd}), note that 
\ifieee
\begin{align*}
\sum_{r=0}^{k-1} u_{kr+i}^* u_{kr+j} 
& = 
\sum_{r=0}^{k-1} D^{-r} S^{-kr-i} S^{kr+j} D^r 
\\&
= 
\sum_{r=0}^{k-1} D^{-r} S^{j-i} D^r\\
& =   
\left(\sum_{r=0}^{k-1} \omega^{-r(j-i)} \right) S^{j-i}. 
\end{align*}
\else
\begin{align*}
\sum_{r=0}^{k-1} u_{kr+i}^* u_{kr+j} 
& = 
\sum_{r=0}^{k-1} D^{-r} S^{-kr-i} S^{kr+j} D^r 
= 
\sum_{r=0}^{k-1} D^{-r} S^{j-i} D^r\\
& =   
\left(\sum_{r=0}^{k-1} \omega^{-r(j-i)} \right) S^{j-i}. 
\end{align*}
\fi
Since $\omega$ is a primitive root of unity, so is $\omega^{-1}$. Thus, 
$\omega^{-(j-i)}$ is a $k$-th root of unity with $\omega^{-(j-i)} \neq 1$.
It follows that $\sum_{r=0}^{k-1} \omega^{-r(j-i)} = 0$. 
Thus, $\sum_{r=0}^{k-1} u_{kr+i}^* u_{kr+j}  = 0$ whenever $i \neq j$. 
On the other hand, 
$$\sum_{r=0}^{k-1} u_{kr+i}^* u_{kr+i}  = \sum_{r=0}^{k-1} I = kI,$$
and (iii) is proved. 

By \cite[Lemma 4]{dsw},
$$\vartheta(\cl S_k\otimes \cl S_{k^2}) \geq \vartheta(\cl S_k) \vartheta(\cl S_{k^2}) = k^3,$$
and the proof is complete.
\end{proof}

\begin{corollary}\label{cor:betterbounds}
  \begin{enumerate}[(i)]
  \item The ratios \ifieee \[\text{\fi$\vartheta(\cl S) / \beta(\cl S)$ and
    $\vartheta(\cl S) / \gamma(\cl S)$\ifieee}\]\else \fi can be arbitrarily large,
    as $\cl S$ varies over all non-commutative graphs.

  \item For $\pi\in\{\beta,\gamma\}$, the ratio $\vartheta(\S)/\pi(\S)$ can be 
  arbitrarily large, as $\S$ varies over all non-commutative graphs with  $\tfrac12 \pi(\S)\leq \Theta(\S)\leq \pi(\S)$.
\end{enumerate}
\noindent Moreover, these statements hold if throughout we replace $\vartheta$ by~$\thetatilde$, the quantum Lov\'asz theta number.
\end{corollary}
\begin{proof} Since $\vartheta\leq \thetatilde$, it suffices to prove these statments for $\vartheta$.

(i) For the first ratio, consider $\S=\cl T_{\beta}:= \cl S_k \otimes \cl S_{k^2}$ and apply Theorem~\ref{th_betabetter}.
For the second ratio, let
$\cl T_{\gamma}$ be the span of the matrices $B_{r,s} \in \cl S_k \otimes \cl S_{k^2}$ appearing in the
proof of Theorem \ref{th_betabetter}.   
Then  $\cl T_{\gamma} \subseteq \cl S_k \otimes \cl S_{k^2}$ is an operator system and the set $\{ B_{r,s} \}$ is one of the terms that appear in the minimum that defines $\gamma(\cl T_{\gamma})$.
Hence,
\[ \gamma(\cl T_{\gamma}) \le \rank\left((B_{r,s})\right) \le k^2 < k^3 \le \vartheta(\cl S_k \otimes \cl S_{k^2}) \le \vartheta(\cl T_{\gamma}).\]

(ii)  For $\pi\in\{\beta ,\gamma\}$, consider $\R_{\pi}:=\T_{\pi}\oplus \bC I_{k^2}\subseteq M_{k^3+k^2}$. For $k>3$, 
by Proposition \ref{prop:elem} (v) and Remark \ref{remark:pi=1} (iv), we have
\ifieee
\begin{align*}
  \pi(\R_{\pi})=\pi(\T_{\pi})+\pi(\bC I_{k^2})
      &\leq k^2+k^2=2k^2 = 2\alpha(\bC I_{k^2})\\&\leq 2\alpha(\R_{\pi})\leq 2\Theta(\R_\pi).
\end{align*}
\else
\begin{align*}
  \pi(\R_{\pi})=\pi(\T_{\pi})+\pi(\bC I_{k^2})
      &\leq k^2+k^2=2k^2 = 2\alpha(\bC I_{k^2})\leq 2\alpha(\R_{\pi})\leq 2\Theta(\R_\pi).
\end{align*}
\fi
Since $\vartheta$ is order-reversing for inclusion of operator systems and $\R_{\pi}\subseteq(\S_{k}\otimes \S_{k^2})\oplus \bC I_{k^2} \subseteq (\S_{k}\oplus \bC)\otimes \S_{k^2}$ (the second inclusion holds up to a unitary shuffle equivalence) 
and it is easy to see that $\vartheta(\cl S \oplus \cl T)\geq \vartheta(\cl S)$ for any operator systems $\cl S$ and $\cl T$, we have
\[\vartheta(\R_{\pi})\geq \vartheta(\S_{k}\otimes \bC)\vartheta(\S_{k^2})\geq k^3.\qedhere\]
\end{proof}

\appendix

\section{Ordering capacity bounds}
\label{sec:order-proofs}

In this appendix, we briefly summarise the order relationships between
various bounds on the quantum Shannon zero-error capacity. These may
be succinctly described by the directed graph in
Figure~\ref{fig:order}. The parameters $\alpha_q$, $\chi_q$ and $\chi$ are
defined in~\cite{dsw}, and~\cite[Definition~11]{stahlke};
the reader should swap $\S$ and $\S^\perp$ when translating between
our non-commutative graphs and Stahlke's ``trace-free non-commutative
graphs''.

\begin{figure}
  \centerline{\includegraphics{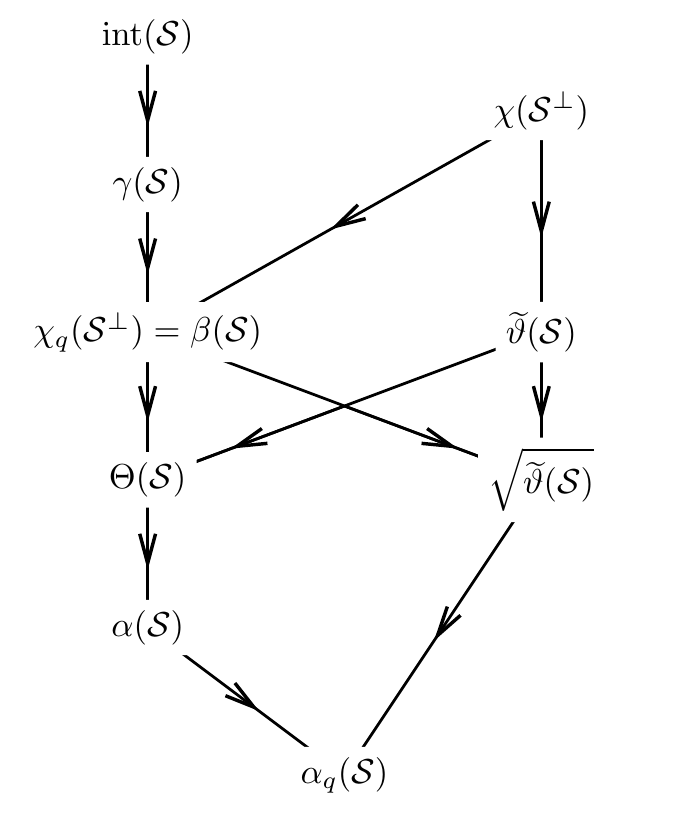}}
  \caption{A directed graph showing the partial order among various parameters bounding the quantum Shannon zero-error capacity $\Theta(\S)$ of a non-commutative graph~$\S$. The ordering $\pi_1(\S)\leq \pi_2(\S)$ for every operator system $\S\subseteq M_n$ is indicated by placing $\pi_1(\S)$ below $\pi_2(\S)$, joined with a path directed towards $\pi_1(\S)$; the absence of a directed path between a pair of vertices indicates that the corresponding parameters are incomparable.
}\label{fig:order}
\end{figure}

Let $\S\subseteq M_n$ be an operator system. As observed in Remark \ref{remark:pi=1} (ii), we
have $\chi_q(\S)=\beta(\S)$. The first two inequalities in the chain
\[\alpha_q(\S)\leq \sqrt{\thetatilde(\S)}\leq \chi_q(\S^\perp)\leq \chi(\S^\perp)\] appear in~\cite{stahlke}, following Corollary~20, and the third inequality is a simple consequence of his Proposition~9. The inequality $\alpha_q(\S)\leq \alpha(\S)$ is immediate from the definitions (and appears in~\cite[Proposition~2]{dsw}), and we have seen 
in Theorems~\ref{thm:inncc} and~\ref{th_boundbeta} that $\alpha\leq\Theta\leq\beta\leq \gamma\leq \inter$. The inequality $\Theta\leq \thetatilde$ follows immediately from~\cite[Proposition~2 and Corollary~10]{dsw}, noting that in the notation of that paper, $\log_2\Theta=C_0\leq C_{0E}$; and $\sqrt{\thetatilde}\leq \thetatilde$ is trivial.

It only remains to prove the incomparability assertions of Figure~\ref{fig:order}. These follow from the inequalities already established and the examples below.
\begin{itemize}
\item Let $G=C_5$ be the $5$-cycle, and let $\S=\S_G$. Lov\'asz has shown~\cite{lo}
  that $\vartheta(G)=\sqrt5$ while, for graph operator systems, as
  pointed out in~\cite{dsw}, we have $\thetatilde(\S_G)=\vartheta(G)$.
It is not difficult to see that 
$$\alpha(\S_G)=\alpha(G)=2<\beta(G)=\beta(\S_G).$$ 
So, in this example,
  \[\sqrt{\thetatilde(\S)}<\alpha(\S) \qand \thetatilde(\S)<\beta(\S).\]
\item Consider $G=C_6^c$, the complement of the $6$-cycle, and $\S=\S_G$. It is easy to see directly that $\gamma(\S)=\gamma(G)>2$, and $\chi(\S^\perp)=\chi(C_6)=2$, so in this case,
  \[ \chi(\S^\perp)<\gamma(\S).\]
\item Let $\S$ be the operator system of Proposition~\ref{prop:n=2,gamma=3} (i.e., in the notation of Section~\ref{sec:applications-to-capacity}, $\S=\S_2$). Note that $\alpha(\S)=1$. We claim that if $\T$ is any operator system with $\alpha(\T)=1$, 
then $\alpha(\S\otimes \T)=1$. Indeed, $\S\otimes \T$ may be identified with all 
$2\times 2$ block matrices of the form $\begin{sbmatrix}T&A\\B&T\end{sbmatrix}$ for $T,A,B\in \T$, and if $x,y$ are non-zero vectors with $xy^*\in (\S\otimes \T)^\perp$, then writing $x=\begin{sbmatrix}x_1\\x_2\end{sbmatrix}$ and $y=\begin{sbmatrix}y_1\\y_2\end{sbmatrix}$, we obtain $xy^*=(x_iy_j^*)_{i,j=1,2}\in (\S\otimes \T)^\perp$. By considering the off-diagonal entries and the condition $\alpha(\T)=1$, 
it readily follows that $x_1 = 0$ or $y_2 = 0$, and $x_2=0$ or $y_1=0$. If $x_1=0$, then $y_1=0$; 
hence, $xy^*=0\oplus x_2y_2^*$, so $x_2y_2^*\in \T^\perp$, so $x=y=0$, a contradiction. 
The other case proceeds to a similar contradiction, so $\alpha(\S\otimes \T)=1$. 
Hence, in particular, $\Theta(\S)=1$. On the other hand, $\thetatilde(\S)=2$ by~\cite[p.~1172]{dsw}; 
thus, in this case we have
\[\Theta(\S)<\sqrt{\thetatilde(\S)}.\]
\item Finally, let $\S=\bC I_2$ to obtain an example for which
  \[ \inter(\S)<\thetatilde(\S),\]
  since the left hand side is~$2$ by Remark~\ref{remark:pi=1}~(iv), and, as observed in~\cite{dsw}, the right hand side is~$4$.
\end{itemize}

\section*{Acknowledgements}
The authors are grateful to the Fields Institute and the Institut Henri Poincar\'e for financial support to attend the Workshop on Operator Systems in Quantum Information and the Workshop on Operator Algebras and Quantum Information Theory, respectively, greatly facilitating our work on this project. The first named author also wishes to thank Helena \v Smigoc and Polona Oblak for stimulating discussion of the minimum semidefinite rank.

\end{document}